\documentclass[10pt,a4paper, fleqn]{article}
\usepackage{etex}
\usepackage[utf8]{inputenc}
\usepackage[TS1,T1]{fontenc}
\usepackage{xspace}

\usepackage[a4paper,tmargin=3truecm,bmargin=3truecm,hmargin=2.5truecm]{geometry}

\usepackage{ifthen}

\usepackage{tikz-cd}
\usetikzlibrary{external}

\usepackage[english]{babel}

\usepackage{amsfonts}
\usepackage{lmodern}

\usepackage{enumitem}
\newlist{enum-hypothesis}{enumerate}{1}
\setlist[enum-hypothesis]{label=(Hyp.~\arabic*),itemsep=0pt, parsep=0pt, labelwidth=5em, leftmargin=5em}

\setlist[enumerate,1]{label=\arabic*., ref=\arabic*, topsep=1pt, itemsep=2pt, parsep=0pt, leftmargin=1.5em, itemindent=0em, labelsep=0.2em, labelwidth=1.3em}
\setlist[enumerate,2]{label=\alph*., ref=\theenumi.\alph*, topsep=1pt, itemsep=2pt, parsep=0pt, leftmargin=0.5em, itemindent=0em, labelsep=0.2em, labelwidth=1.5em}
\setlist[enumerate,3]{label=\roman*., ref=\theenumii.\roman*, topsep=1pt, itemsep=2pt, parsep=0pt, leftmargin=0.5em, itemindent=0em, labelsep=0.2em, labelwidth=1.2em}

\setlist[itemize,1]{topsep=1pt, itemsep=2pt, parsep=0pt}

\usepackage{array}   
\newcolumntype{R}{>{\raggedleft\arraybackslash$}p{1.5em}<{$}} 
\usepackage{booktabs, dcolumn}
\usepackage{graphicx, subfig}

\usepackage{pbox}

\usepackage{amsmath,amssymb,mathtools,slashed,mathdots}
\usepackage{nccmath}
\usepackage[thmmarks,amsmath]{ntheorem}
\usepackage{bm}

\newtheorem{theorem}{Theorem}[section]
\newtheorem{proposition}[theorem]{Proposition}
\newtheorem{lemma}[theorem]{Lemma}

\newtheorem{assumption}[theorem]{Assumption}

\newtheorem{definition}[theorem]{Definition}

\theoremsymbol{$\square$}
\theoremstyle{plain}
\theorembodyfont{\upshape}

\theoremsymbol{$\Diamond$}

\theoremsymbol{}
\theoremstyle{break}
\theorembodyfont{\slshape}

\theoremheaderfont{\scshape}
\theorembodyfont{\upshape} 
\theoremstyle{nonumberplain}
\theoremseparator{} 
\theoremsymbol{\rule{1.3ex}{1.3ex}} 
\newtheorem{proof}{Proof}

\DeclareSymbolFont{largesymbols}{OMX}{cmex}{m}{n}

\usepackage[square,numbers,compress,merge]{natbib}
\usepackage{hyperref}

\hypersetup{
plainpages=false,
colorlinks=true,
linkcolor=black, 
anchorcolor=black, 
citecolor=black, 
urlcolor=black, 
menucolor=black, 
filecolor=black, 
bookmarksopen=true,
bookmarksnumbered=true}

\hypersetup{
pdfauthor={Mathilde Guillaud, Serge Lazzarini, Thierry Masson, Guillaume ?},
pdftitle={???},
pdfsubject={},
pdfcreator={pdflatex},
pdfproducer={pdflatex},
pdfkeywords={}
}

\newcommand{\bbA}{\mathbb{A}}

\newcommand{\bbC}{\mathbb{C}}

\newcommand{\bbR}{\mathbb{R}}

\newcommand{\bbbone}{{\text{\usefont{U}{bbold}{m}{n}\char49}}}
\newcommand{\bbbzero}{{\text{\usefont{U}{bbold}{m}{n}\char48}}}

\newcommand{\calA}{\mathcal{A}}

\newcommand{\calD}{\mathcal{D}}
\newcommand{\calE}{\mathcal{E}}
\newcommand{\calF}{\mathcal{F}}
\newcommand{\calG}{\mathcal{G}}

\newcommand{\calI}{\mathcal{I}}

\newcommand{\calO}{\mathcal{O}}

\newcommand{\calZ}{\mathcal{Z}}

\newcommand{\kg}{\mathfrak{g}}

\newcommand{\ksu}{\mathfrak{su}}
\newcommand{\ku}{\mathfrak{u}}
\newcommand{\kv}{\mathfrak{v}}

\newcommand{\bD}{\mathbf{D}}
\newcommand{\bE}{\mathbf{E}}
\newcommand{\bM}{\mathbf{M}}

\newcommand\hu{\hat{u}}

\newcommand{\Tphya}{\tilde{\phya}}
\newcommand{\Talpha}{\tilde{\alpha}}
\newcommand{\Tsigma}{\tilde{\sigma}}
\newcommand{\Trho}{\tilde{\rho}}
\newcommand{\Txi}{\tilde{\xi}}
\newcommand{\Tgamma}{\tilde{\gamma}}

\newcommand{\phyA}{\mathsf{A}}
\newcommand{\phya}{\mathsf{a}}
\newcommand{\hphyA}{\widehat{\mathsf{A}}}
\newcommand{\phyF}{\mathsf{F}}
\newcommand{\phyf}{\mathsf{f}}
\newcommand{\phyD}{\mathsf{D}}

\newcommand{\smallpmatrix}[1]{\left( \begin{smallmatrix} #1 \end{smallmatrix} \right)}
\newcommand{\defeq}{\vcentcolon=} 
\newcommand{\rdefeq}{=\vcentcolon} 

\DeclareMathOperator{\Ad}{Ad}

\DeclareMathOperator{\Det}{Det}

\DeclareMathOperator{\Id}{Id}

\DeclareMathOperator{\Lie}{Lie}  

\DeclareMathOperator{\tr}{tr}	   
\DeclareMathOperator{\Tr}{Tr}	   

\DeclarePairedDelimiter\abs{\lvert}{\rvert}
\DeclarePairedDelimiter\norm{\lVert}{\rVert}

\newcommand{\fsG}{\underline{G}} 
\newcommand{\fsA}{\underline{\bbA}} 
\newcommand{\fsE}{\underline{E}} 
\newcommand{\fsF}[1][]{%
\ifthenelse{\equal{#1}{}}%
{\underline{F}}%
{\underline{F_{#1}}}%
} 
\newcommand{\fsV}{\underline{V}} 
\newcommand{\fsbbC}{\underline{\bbC}}
\newcommand{\fsbbR}{\underline{\bbR}}
\newcommand{\fsbbRps}{\underline{\bbR^\ast_+}}
\newcommand{\fsU}[1]{\underline{U(#1)}}
\newcommand{\fsSU}[1]{\underline{SU(#1)}}
\newcommand{\fsku}[1]{\underline{\ku(#1)}}

\newcommand{\lieG}{\kg} 
\newcommand{\ggG}{\calG} 
\newcommand{\invG}{\calI^G} 
\newcommand{\connA}{\calA} 
\newcommand{\invA}{\calI^{\bbA}} 
\newcommand{\fieldsE}{\calE} 
\newcommand{\fieldsF}{\calF} 
\newcommand{\invF}[1][]{%
\ifthenelse{\equal{#1}{}}%
{\calI^F}%
{\calI^{F_{#1}}}%
} 
\newcommand{\invE}{\calI^E} 
\newcommand{\dress}{\calD} 
\newcommand{\udress}{{\overline{\calD}}} 

\newcommand{\spaceupsi}{\mathbf{F}} 

\newcommand{\FC}{\text{\textup{C}}} 
\newcommand{\GT}{\text{\textup{GT}}} 
\newcommand{\GA}{\text{\textup{GA}}} 
\newcommand{\DC}{\text{\textup{DC}}} 
\newcommand{\UDC}{\text{\textup{UDC}}} 

\newcommand{\GF}{\text{\textup{F}}} 
\newcommand{\GFM}{{\widehat{\GF}}} 
\newcommand{\GFDC}{{\widetilde{\GF}}} %

\newcommand{\DCGFM}{\DC\GFM} %

\newcommand{\dd}{\text{\textup{d}}}
\newcommand{\fsdd}{\text{\textbf{d}}}

\newcounter{mnotecount}[section]
\renewcommand{\themnotecount}{\thesection.\arabic{mnotecount}}
\newcommand{\mnote}[1]%
{\protect{\stepcounter{mnotecount}}${}^{\text{\footnotesize$\bullet$\themnotecount}}$%
\reversemarginpar%
\marginpar{\raggedleft\footnotesize$\bullet$\themnotecount: #1}}

\newlength{\mnotewidth}
\numberwithin{equation}{section}
\allowdisplaybreaks

\begin{document}
\renewcommand\figurename{Fig.}


{
\makeatletter\def\@fnsymbol{\@arabic}\makeatother 
\title{Gauge Fixing in QFT and the Dressing Field Method}
\author{Mathilde Guillaud, Serge Lazzarini, Thierry Masson\\
{\small Centre de Physique Théorique}%
\\
\small{Aix Marseille Univ, Université de Toulon, CNRS, CPT, Marseille, France}\\[2ex]
}
\date{version: June 28, 2024}

\maketitle
}

\begin{abstract}
In this paper, we revisit the Dressing Field Method (DFM) in the context of Quantum (Gauge) Field Theories (QFT). In order to adapt this method to the functional path integral formalism of QFT, we depart from the usual differential geometry approach used so far to study the DFM which also allows to tackle the infinite dimension of the field spaces. Our main result is that gauge fixing is an instance of the application of the DFM. The Faddeev-Popov gauge fixing procedure and the so-called unitary gauge are revisited in light of this result.
\end{abstract}


\tableofcontents

\newpage
\section{Introduction}
\label{sec introduction}

The Dressing Field Method (DFM) was introduced in \cite{FourFranLazzMass14a} as a way to reduce gauge degrees of freedom in gauge field theories as a change of variables among the fields of the theory. Since then, many applications of this method have been proposed, in different contexts by collecting examples (some of them coming from the literature, see \cite{AttaFranLazzMass18a} for a review), but always in relation to \emph{classical} gauge field theories. This is why, until now, this method was only considered in the framework of differential geometry, which is the natural one for classical gauge field theories. 

Let us just recall that to apply the DFM, one has to select in the gauge model a  (group valued) field $u$, the dressing field, which supports a specific gauge transformation: $u$ must be constructed using (part/some of the) degrees of freedom \emph{in} the model, so that it is not an external element of the model. Then the dressing field is used to “dress” all the gauge and the matter fields in the model with relations which look like gauge transformations (but they are not!). This produces dressed fields with less (and even no more in the best case scenario) gauge variance. The classical examples studied so far show that dressed fields are \emph{composite fields} while keeping the locality principle.

\smallskip
In this paper, we would like to start the study of applications of the DFM at the quantum level, in the functional approach to Quantum Field Theories (QFT). The first application will focus on the Faddeev-Popov gauge fixing procedure (FPGFP) in the functional integral, whose purpose, as the one of the DFM, is to get rid of gauge degrees of freedom (Section~\ref{sec FPGFP revisited}). The FPGFP relies on the choice of a representative in each gauge class of fields, while the DFM makes apparent gauge invariant fields. 

The main result of this paper is that, in the FPGFP, \emph{the Gauge Fixing Procedure turns out to be an instance of the DFM}. In short: for \emph{ideal} gauge fixing maps (see Definition~\ref{def ideal gauge fixing}), the transformation occurring in the FPGFP turns out to be a dressing composition, and not a gauge transformation as usually claimed. This result is proved using our natural Assumption~\ref{assum equivariant map ideal}. Upon using this result, we rewrite the FPGFP in the framework of the DFM, taking into account the subtleties of the FPGFP and the special features of the DFM, in particular concerning gauge invariant fields.

In the recent paper \cite{BergFran24a}, the conclusion that the gauge fixing procedure is an instance of the DFM is also drawn for a $U(1)$ model in the Lorenz gauge. We refer to this paper for bibliographical comments about the comparison between the DFM and the gauge fixing procedure.

One important consequence of this result is the possibility to compare different gauge fixing conditions by looking at their associated dressing fields in the same functional space. Indeed, the dressing field $u$ is constructed out of the fields contained in the model as expected by the method, but it also uses (as expected in relation to the FPGFP) the extra ingredient which is the gauge fixing condition. For instance, this allows us to relate the $R_\xi$ gauge fixing condition to the “unitary gauge” fixing condition by taking the limit $\xi \to \infty$ \emph{at the level of dressing fields $u$ themselves}. It is worthwhile to notice that there is no consensus that unitary gauges are true gauge fixings, see for instance  \cite{Wein73a} for one viewpoint and \cite{DolanJackiw} for the other one.  However,  several examples of “unitary gauges”, for instance in the Standard Model of Particle Physics (SMPP)~\cite{MassWall10a,AttaFranLazzMass18a, FourFranLazzMass14a}, can be understood as an application of the DFM.  The above mentioned limit amounts to considering that all these “gauge fixing conditions” ($R_\xi$ and unitary gauges) fall into the unifying standpoint of the DFM. 

One key feature of many examples of the DFM studied so far is that the dressing field $u$ \emph{is local} in the fields in the model (in the usual sense of QFT). However, for many gauge fixing conditions (Lorenz, $R_\xi$), we can observe that the dressing field $u$ is \emph{not} local in the fields in the model. This criteria of locality allows us to set the “unitary gauges” apart from these gauge fixing conditions. It is already known that the “unitary gauges” are of major interest because they show the observed degrees of freedom. Following the (philosophical) line of reasoning developed in \cite{Fran19a} (see also \cite{BergFran24a}) about the locality of $u$ in terms of the fields in the model, we make the assumption that the locality of the dressing field is related to the observability of the dressed fields. It is out of the scope of this paper to address this point further. 

\smallskip
As explained in details in Section~\ref{sec the framework}, in the present paper we will not use the usual fiber bundle approach to gauge field theories. Until now, the dressing field method has been developed and illustrated in that framework since we focused mainly on classical field theories. But this is not the most pertinent framework for the functional approach to QFT, even if it can be very useful for specific problems. For instance, the geometrical structures are certainly not the best tools to use in the functional integral of the quantization procedure.

So, for the applications we have in mind, especially the relation between the DFM and the FPGFP, we have to adapt the DFM to the usual tools devoted to this procedure. This is why, in this paper, we rewrite the dressing field method in a more flexible framework, based, on one hand, on \emph{functional spaces}, that is (smooth) maps on space-time (or locally on space-time) with values in some spaces (Lie group, representation vector spaces for these Lie groups…), and, on the other hand, on the \emph{gauge group} defining the gauge model 
under study. 

In order to characterize “gauge fields”,\footnote{In this paper, “gauge fields” collectively refers to all the fields in the theory on which the gauge group acts.} we will then equip these functional spaces with actions of the gauge group. These functional spaces endowed with such an action will be called \emph{field spaces}. It is worthwhile to notice that different field spaces can be based on the same underlying functional space, but with different actions of the gauge group: this will play a key role in our approach. Examples of such spaces are provided in Section~\ref{sec the framework}, where the relation to the usual approach in terms of fiber bundles and connections is explained. 

We will also introduce maps between these functional/field spaces, in order to get a general framework to write the DFM using such instances of maps. Especially, in Section~\ref{sec field-composer and the dressing field method}, we put forward the concept of “Field-Composer”, which can be used at many places in relation to the DFM. This shows in particular that the DFM can be naturally conceived in the above mentioned framework of functional spaces and actions of the gauge group defined on them.

\medskip
Many computations given in the paper may look “usual” on first reading. But, as mentioned in several papers now (see the review \cite{AttaFranLazzMass18a} for all the details and references therein), the DFM is close in many respect, but not equivalent, to the ordinary methods used so far to reduce gauge symmetries. It was already noticed that it can “replace” the Spontaneous Symmetry Breaking Mechanism (SSBM) in the SMPP, opening some new avenues for understanding the Electro-Weak sector of the SMPP (since it decouples the apparition of the observed degrees of freedom from the choice of an energy scale at which to produce mass terms). In the present paper, we open a new chapter by relating the DFM to the FPGFP that was thought to be quite different before the present work (even by the authors).

Let us give a simple illustration of the fact that the DFM provides highly satisfactory responses to some usual questions related to gauge fixing. To do that, let us apply the DFM to the simple example of an Abelian $U(1)$ toy model defined by the Lagrangian (we use notations introduced in Section~\ref{sec rxi gauge fixing and unitary gauge})
\begin{align}
\label{L-U(1)}
L[\phyA, \phi] \defeq 
[(\partial_\mu - i e \phyA_\mu) \phi]^\dagger [(\partial^\mu - i e \phyA^\mu)\phi] - V(\phi) - \tfrac{1}{4} \phyF_{\mu\nu} \phyF^{\mu\nu}
\end{align}
where $\phi$ is $\bbC$-valued, $\phyF_{\mu\nu} $ is the field strength tensor associated to $\phyA_\mu$, $V(\phi) = \tfrac{\mu^2}{2} \phi^\dagger \phi + \tfrac{\lambda}{4} (\phi^\dagger \phi)^2$, and the actions of a gauge transformation with $\gamma = e^{i \alpha} \in \fsU{1}$ ($U(1)$-valued smooth map) are $\phi^\gamma \defeq \gamma^{-1}\phi$ and $\phyA_\mu^\gamma = \phyA_\mu + \tfrac{i}{e} \gamma^{-1} \partial_\mu \gamma = \phyA_\mu - \tfrac{1}{e} \partial_\mu \alpha$. Let $\phi = \rho e^{i \chi}$ with $\rho \defeq \abs{\phi}$, so that under the gauge transformation $\gamma$ one has $\rho^\gamma = \rho$ and $\chi^\gamma = \chi - \alpha$. The Lagrangian can be written in the $(\rho, \chi)$ field variables:
\begin{align*}
L[\phyA, \rho, \chi] = 
(\partial_\mu \rho) (\partial^\mu \rho) + \rho^2 (\partial_\mu \chi - e \phyA_\mu)(\partial^\mu \chi - e \phyA^\mu) - V(\rho) - \tfrac{1}{4} \phyF_{\mu\nu} \phyF^{\mu\nu}
\end{align*}
The purpose of the usual gauge fixing procedure for the so-called “unitary gauge” is to remove any occurrence of the $\chi$ field. To do that, the idea is to perform a gauge transformation with $\gamma$ such that $\alpha = \chi$. But, for any gauge transformation $\rho \mapsto \rho$, $\chi \mapsto \chi - \alpha$, $\phyA_\mu \mapsto \phyA_\mu - \tfrac{1}{e} \partial_\mu \alpha$, the expression $\partial_\mu \chi - e \phyA_\mu$ transforms into itself (as expected). So, there is no gauge transformation that can remove the $\chi$ field.\footnote{It is customary that only a “partial” gauge transformation with $\alpha = \chi$ applied only to the fields $\phyA_\mu$, but not to the field $\chi$, could do the job. This is clearly not a satisfactory procedure.}

The DFM is strongly related to this line of reasoning and its success, for the same problem, relies on the fact that it considers the right objects in the right spaces, and interprets some usual relations in a different manner (gauge transformations for instance).

The first step consists in identifying in the model the dressing field $u$ which takes its values in $U(1)$ and which transforms as $u^\gamma = \gamma^{-1} u$. With the previous notation, a natural candidate for $u$ is $u = e^{i\chi}$, that is, we write $\phi = \rho u$, so that $u$ is a local expression in terms of the components of $\phi$. Here, we see that $u$ looks very much like the $\gamma$ proposed in the unitary gauge fixing procedure. The second step of the method is to dress all the gauge fields with $u$, using the usual relations for the action of the gauge group, but with $u$ instead of $\gamma$. Here again, it looks like we perform a gauge transformation on all the fields. But, as explained in detail in \cite{AttaFranLazzMass18a, FourFranLazzMass14a}, \emph{the dressing field $u$ is not an element of the gauge group} so that \emph{the dressing of all the fields by $u$ can not be a gauge transformation} (it is a redistribution of the degrees of freedom in new field variables). The dressed fields for the $\phyA_\mu$'s are the fields $\phya_\mu \defeq \phyA_\mu + \tfrac{i}{e} u^{-1} \partial_\mu u$ and the dressed field for $\phi$ is $\rho$. Since this change of variables in the space of fields is invertible, one can write the Lagrangian in terms of these dressed fields:
\begin{align*}
L[\phya, \rho] \defeq 
[(\partial_\mu - i e \phya_\mu) \rho]^\dagger [(\partial^\mu - i e \phya^\mu)\rho] - V(\rho) - \tfrac{1}{4} \phyf_{\mu\nu} \phyf^{\mu\nu}
\end{align*}
where $\phyf_{\mu\nu}$ has the same expression in terms of the $\phya_\mu$'s as $\phyF_{\mu\nu}$ in terms of the $\phyA_\mu$'s. In this Lagrangian, the $\chi$ field has disappeared as desired. Notice that a change of field variables yields a Jacobian in the functional integral. Two examples of such Jacobians are computed in Appendix~\ref{sec functional differential and jacobians}.

One way to understand why the procedure works with the DFM but not with the gauge transformation is to remember that the gauge transformation defined by $\gamma$ cannot change the status of the objects, in particular the fields $\phyA_\mu$, which still define a \emph{connection} $1$-form. By definition, a gauge transformation preserves field spaces (since a field space is precisely defined to support a specific action of the gauge group, see Section~\ref{sec the framework}). On the contrary, in this example at hand, the dressing field $u$ in the DFM, which captures the same degrees of freedom as $\gamma$, amounts to defining objects (the dressed fields) belonging to \emph{new field spaces}. The fields $\phya_\mu$ no longer define a connection $1$-form since they form a gauge invariant object (they belong to a field space supporting the trivial action of the gauge group, see the notion of Field-Composer in Section~\ref{sec field-composer and the dressing field method}). In the terminology to be defined in Section~\ref{sec the framework}, $\gamma$ and $u$ belong to the same functional space, as $U(1)$-valued functions, while they do not belong to the same field spaces since they do not support the same action of the gauge group. It is the same for the functions $\phyA_\mu$ and $\phya_\mu$. So, by its very definition, a gauge transformation cannot hide the field $\chi$ (invariance of the combination $\partial_\mu \chi - e \phyA_\mu$), while the approach of the DFM is to “compose” the $\phyA_\mu$'s and $\chi$ functions into the new fields $\phya_\mu$. This is why $\chi$ disappears in the dressed Lagrangian, as part of the $\phya_\mu$'s.

\section{The Framework}
\label{sec the framework}

The usual modern mathematical approach to (classical) gauge fields makes use of fiber bundles. Here, as explained in the Introduction, we will not use this framework, since we will only consider \emph{local} fields (on the space-time manifold). Indeed, one of the main results concerning the DFM, \cite[Prop.~2]{FourFranLazzMass14a}, tells us that the existence of a \emph{global} dressing field with values in the whole structure group implies the triviality of the principal fiber bundle.\footnote{In the paper, we focus ourselves on the whole structure group and not to possible subgroups.} So, instead of relying on fiber bundles to identify the field spaces, we will rely on \emph{the action of the gauge group on local fields} defined on open subsets $U$ of the $m$-dimensional space-time manifold $M$. Working with such local fields will circumvent the global triviality constraint and permit to make direct contact with the structures used in functional integrals of QFT. Notice that $U$ can be $M$ itself: in QFT, one has $M = \bbR^4$ and all the fiber bundles are trivial (contractive space) so that one can take $U = \bbR^4$.

\subsection{Functional Spaces, Field Spaces, and Gauge Group Actions}

Let us denote by $G$ the structure group of our model, with Lie algebra $\lieG$. For any open subset $U$ of $M$ and any representation vector space $E$ of $G$, let us introduce the following local \emph{functional spaces}:
\begin{align*}
\fsG_U &\defeq \{ g : U \to G \},
&
\fsA_U &\defeq \{ a = (a_\mu) \ /\  a_\mu : U \to \lieG \},
&
\fsE_U &\defeq \{ \varphi : U \to E \},
\end{align*}
where all the maps are smooth. The space $\fsG_U$ is a group when equipped with the natural group law inherited from the group law of $G$ and, in the same way, $\fsE_U $ is a vector space. We emphasize that these spaces are equipped only with their functional space structure (which depends on the target space, and on which topological structures could be added, but this is outside the scope of this paper). The main point of our approach is that these spaces will be equipped with different actions of the gauge group.

Let us then first define the \emph{local gauge group} $\ggG_U$ as follows: it is $\fsG_U$ as a group (and so as a functional space), equipped with the right action of (the group) $\fsG_U$ defined by $\gamma^g \defeq \alpha_g(\gamma) \defeq g^{-1} \gamma g$ for any $g \in \fsG_U$ and $\gamma \in \ggG_U$. It is important to distinguish the two mathematical structures: $\fsG_U$ is a group, and $\ggG_U$ is a group equipped with an action of the group $\fsG_U$. Notice that this action induces, with the same formula, an action of the group $\ggG_U$ on itself. This is this action that we will consider in the following.

As pointed out before, we equip now some functional spaces with right actions of the group $\ggG_U$, and we call them \emph{field spaces}. It will be important to remember that different field spaces can have the same underlying functional space, since the actions can be different. The first field space at hand is $\ggG_U$ for which the functional space is  $\fsG_U$ equipped with the above action $\alpha$. We will use special notations for the following field spaces:\footnote{Wherever possible, we will also try to use different notations for the elements of these spaces}
\begin{itemize}
\item The \emph{field space of (local) connections} $\connA_U$ is the functional space $\fsA_U$ equipped with the action $A \mapsto A^\gamma \defeq \gamma^{-1} A \gamma + \gamma^{-1} \dd \gamma$ for any $\gamma \in \ggG_U$ and $A \in \connA_U$.

\item The \emph{field space of $E$-valued fields} $\fieldsE_U$ is the functional space $\fsE_U$ equipped with the action $\phi \mapsto \phi^\gamma \defeq \ell_{\gamma^{-1}} \phi$ for any $\gamma \in \ggG_U$ and $\phi \in \fieldsE_U$, and where $\ell$ is the representation of $G$ on $E$ (\textit{i.e.} a left action).

\item The \emph{field space of invariant connections} $\invA_U$ is the functional space $\fsA_U$ equipped with the trivial action $B \mapsto B^\gamma \defeq B$ for any $\gamma \in \ggG_U$ and $B \in \invA_U$.

\item The \emph{field space of invariant $E$-valued fields} $\invE_U$ is the functional space $\fsE_U$ equipped with the trivial action $\psi \mapsto \psi^\gamma \defeq \psi$ for any $\gamma \in \ggG_U$ and $\psi \in \invE_U$.

\item The \emph{dressing field space} $\dress_U$ is the functional space $\fsG_U$ equipped with the action $u \mapsto u^\gamma \defeq \gamma^{-1} u$ for any $\gamma \in \ggG_U$ and $u \in \dress_U$.

\item The \emph{undressing field space} $\udress_U$ is the functional space $\fsG_U$ equipped with the action $v \mapsto v^\gamma \defeq v \gamma$ for any $\gamma \in \ggG_U$ and $v \in \udress_U$.
\end{itemize}

\medskip
Let us explain how these definitions are related to the usual approach on gauge field theories using principal bundles and associated bundles. Indeed, our present approach can be considered as a local version of this usual approach and the previous definitions are obviously strongly related to this approach. 

Let $P = P(M,G)$ be a $G$-principal bundle over the (space-time) base manifold $M$, and let $F$ be a space equipped with a left action of $G$ denoted by $(f,g) \mapsto \rho(g)f$ for any $f \in F$ and $g \in G$. Then the space of (smooth) sections of the associated fiber bundle $P \times_\rho F$ is isomorphic to the space of (smooth) equivariant maps $\phi : P \to F$ satisfying $\phi(p \cdot g) = \rho(g^{-1}) \phi(p)$ for any $p \in P$ and $g \in G$, where $p \cdot g$ is the right action of $G$ on $P$. It is well-known that the gauge group $\ggG$ of $P$ is isomorphic with the space of sections of the associated bundle $P \times_\alpha G$ for the action $\alpha_g(\gamma) = g^{-1} \gamma g$ defined above. Let us denote by $\Psi : P \to G$ a generic element of the gauge group considered as an equivariant map (for the $\alpha$ action) $P \to G$. Then, with previous notations and identifications, the gauge group action $\phi \mapsto \phi^\Psi$ on sections of $P \times_\rho F$ takes the form $\phi^\Psi(p) \defeq \rho(\Psi(p)^{-1}) \phi(p)$. 

Let $U \subset M$ be an open subset such that $P_{\mid U} \simeq U \times G$ and let $s : U \to P_{\mid U}$ be a trivializing section. For any equivariant map $\phi : P \to F$, define its local section $\varphi \defeq s^\ast \phi$ over $U$ and let $\gamma \defeq s^\ast(\Psi)$. Then the gauge group action \emph{at the level of local sections} takes the form $\varphi \mapsto \varphi^\gamma$ with $\varphi^\gamma = \rho(\gamma^{-1}) \varphi$. In particular, the action of the gauge group on itself takes the form presented above. In the same way, we recover the action on (local) connections $A \in \connA_U$.

As expected, there is then a strong relation between the expression of the action of the gauge group on local fields and the field space in which these local fields belong, since the action determines $\rho$, which in turn determines the associated fiber bundle. For instance, let us consider the dressing field space $\dress$. The left action of $G$ to consider is the left multiplication on $G$, $L_g(g') = g g'$, considered as an action of $G$ (group) on $G$ (fiber). Then, a dressing field is a local section of the associated fiber bundle $P \times_L G$, and it is well-known that $P \times_L G \simeq P$. Since a global section of $P$ can only exist if and only if $P$ is trivial, we can not expect such sections (dressing fields) to be globally defined except in the trivial situation $P = M \times G$. But at the local level, local dressing fields can always be considered.

So, working at the local level (over $U$ for which $P_{\mid U}$ is trivial) amounts to considering “local sections” which are always well-defined, and identifying the actions of the gauge group allows to understand the global geometric structures to which these fields (should) belong. This is why in this paper we have chosen to consider gauge fields through this approach. In particular, we will not take interest in the “changes of trivialization”, which are the usual way to identify the bundle structure on which the fields live. Our main focus is on the actions of the gauge group, considered itself as a field space of local sections.

\smallskip
To simplify the presentation and when the open subset $U$ is fixed, we will omit it in the notations.

\subsection{Field-Composer and the Dressing Field Method}
\label{sec field-composer and the dressing field method}

In \cite{FourFranLazzMass14a}, we used a lot the notion of “composite fields”. We would like to clarify its meaning in light of the current approach. The formal definitions and developments presented below may seem cumbersome at first sight, but they are in fact quite useful (and almost necessary) for correctly interpreting the various structures involved in the gauge-fixing process. Moreover, these structures are proving useful and efficient for carrying out certain calculations.

In the following, we will use generic notations for functional spaces and field spaces. Let $\fsF$ be a functional space. Denote by $\fieldsF$ (resp. $\invF$) this functional space equipped with an \textit{a priori} non trivial action of $\ggG$ (resp. with the trivial action of $\ggG$). 
This non trivial action (physically interesting and useful) is extracted from the gauge field model at hand.  $\fieldsF$ is then the usual space for a gauge field, for instance, $\connA$ or $\fieldsE$ given above. In contrast, the trivial action defining $\invF$ will arise from the DFM.
Hence, the field space $\fieldsF$ is the cornerstone of forthcoming constructions.

When necessary, field spaces of type $\fieldsF$ will be distinguished by lower indices.
\begin{definition}
Let $\fieldsF_1, \dots, \fieldsF_{r+1}$ be some generic field spaces on which the actions of the gauge group are denoted by $\fieldsF_i \ni \phi_i \mapsto \phi_i^\gamma$ for any $\gamma \in \ggG$. A \emph{field-composer} is a map $\FC : \fieldsF_1 \times \cdots \times \fieldsF_{r} \to \fieldsF_{r+1}$ which is local in term of fields and satisfies the $\ggG$-equivariance
\begin{align*}
\FC(\phi_1^\gamma, \dots, \phi_r^\gamma) = \FC(\phi_1, \dots, \phi_r)^\gamma
\end{align*}
for any $\phi_i \in \fieldsF_i$ and $\gamma \in \ggG$. 
\end{definition}
Note that for $r=1$, a field-composer is just a $\ggG$-equivariant map between two field spaces.

\smallskip
Let us write $\FC(\phi_i) = \FC(\phi_1, \dots, \phi_r)$. Then one has $\FC(\phi_i^{\gamma_1 \gamma_2}) = \FC(\phi_i^{\gamma_1})^{\gamma_2} = \FC(\phi_i)^{\gamma_1 \gamma_2}$ for any $\gamma_1, \gamma_2 \in \ggG$ since $\phi_i^{\gamma_1 \gamma_2} = (\phi_i^{\gamma_1})^{\gamma_2}$.

Recall that the locality of $\FC$ means that the value of $\FC(\phi_i)$ at any (space-time) point depends only on the values at that point of the fields $\phi_i$ and a finite number of their derivatives.

\medskip
Using the generic notations, let $\GA : \fsF \times \fsG \to \fsF$ be the “Gauge Action Transformation” map which associates to $(\varphi, g) \in \fsF \times \fsG$ the element in $\fsF$ which \emph{would formally correspond to the gauge action of $g$ on $\varphi$ if $g$ were in $\ggG$ and $\varphi$ in $\fieldsF$} (the field space equipped with a non trivial action of $\ggG$).  Since $\GA$ is the functional expression of a right action, we have
\begin{align*}
\GA(\GA(\varphi, g_1), g_2) = \GA(\varphi, g_1 g_2)
\end{align*}
for any $\varphi \in \fsF$ and $g_1, g_2 \in \fsG$.

\begin{proposition}[Declinations of $\GA$ as field-composers]
\label{prop GT-DC-UDC}
For the two declinations of $\fsF$ as field spaces $\fieldsF$ and $\invF$ together with the three declinations of $\fsG$ as field spaces $ \ggG$,  $\dress$, and $\udress$, the gauge action transformation map $\GA$ induces the only three field-composers $\GT : \fieldsF \times \ggG \to \fieldsF$ (“Gauge Transformation”), $\DC : \fieldsF \times \dress \to \invF$ (“Dressing Composer”) and $\UDC : \invF \times \udress \to \fieldsF$ (“Un-Dressing Composer”).
\end{proposition}

\begin{proof}
For any $\phi \in \fieldsF$ and $\gamma, \gamma' \in \ggG$, one has $\GA(\phi^{\gamma'}, \gamma^{\gamma'}) = \GA(\GA(\phi, \gamma'), \gamma'^{-1} \gamma \gamma') = \GA(\phi, \gamma' \gamma'^{-1} \gamma \gamma') = \GA(\phi, \gamma \gamma') = \GA(\GA(\phi, \gamma), \gamma')$, so that $\GT(\phi^{\gamma'}, \gamma^{\gamma'}) = \GT(\phi, \gamma)^{\gamma'}$.

For any $\phi \in \fieldsF$, $u \in \dress$ and $\gamma \in \ggG$, one has $\GA(\phi^\gamma, u^\gamma) = \GA(\GA(\phi, \gamma), \gamma^{-1} u) = \GA(\phi, \gamma \gamma^{-1} u) = \GA(\phi, u)$, so that $\DC(\phi^\gamma, u^\gamma) = \DC(\phi, u) = \DC(\phi, u)^\gamma$.

For any $\phi \in \invF$, $v \in \udress$ and $\gamma \in \ggG$, one has $\GA(\phi^\gamma, v^\gamma) = \GA(\phi, v \gamma) = \GA(\GA(\phi, v), \gamma)$, so that $\UDC(\phi^\gamma, v^\gamma) = \UDC(\phi, v)^\gamma$.

It is easy to check that these three field-composers are the only ones we can construct with the proposed field spaces.
\end{proof}

Notice that since $\ggG$ is a group, $\GT$ inherits the relation $\GT(\GT(\phi, \gamma), \gamma') = \GT(\phi, \gamma \gamma')$ from $\GA$. The proof of the following proposition is straightforward.

\begin{proposition}
\label{prop iota mu FC}
The inverse map $\iota : \fsG \to \fsG$, $\iota(g) \defeq g^{-1}$, induces three field-composers ($\ggG$-equivariant maps) $\iota : \ggG \to \ggG$, $\iota : \dress \to \udress$, and $\iota : \udress \to \dress$.

The multiplication map $\mu : \fsG \times \fsG \to \fsG$, $\mu(g_1, g_2) \defeq g_1 g_2$, induces the five field-composers $\mu : \ggG \times \ggG \to \ggG$, $\mu : \ggG \times \dress \to \dress$, $\mu : \udress \times \ggG \to \udress$, $\mu : \dress \times \udress \to \ggG$, and $\mu : \udress \times \dress \to \invG$ where $\invG$ is the functional space $\fsG$ equipped with the trivial action of $\ggG$.\footnote{We restrict ourselves to the three field spaces $\ggG, \dress$ and $\udress$ as source spaces.}
\end{proposition}

\smallskip
From these properties, we see that the group structure of the functional space $\fsG$ can be lifted to a group structure on the field space $\ggG$ (\textit{i.e.} a group law that is compatible with the action of $\ggG$ on itself). 
From now on, the maps $\mu$ and $\iota$ will be dropped out to the benefit of their respective realisation.

From Prop.~\ref{prop iota mu FC}, for any two dressing fields $u_1, u_2 \in \dress$, there is a unique $\gamma \defeq  \mu(u_1,\iota(u_2)) = u_1 u_2^{-1} \in \ggG$ such that $u_2 = \gamma^{-1} u_1$. This implies that $\dress$ has a unique orbit for the (free) right action of $\ggG$ on $\dress$. A similar result holds for the right action of $\ggG$ on $\udress$.

\begin{lemma}
\label{lemma DCu and UDCv}
For any $u \in \dress$ (resp. $v \in \udress$), the dressing field map $\DC_u : \fieldsF \to \invF$ (resp. the undressing field map $\UDC_v : \invF \to \fieldsF$)  defined by $\DC_u(\phi) \defeq \DC(\phi, u)$ (resp. $\UDC_v(\psi) \defeq \UDC(\psi, v)$) is an isomorphism, but is not $\ggG$-equivariant. Explicitly, for any $\gamma \in \ggG$, one has $\DC_u(\phi^\gamma) = \DC_{\gamma u}(\phi)$ (resp. $\UDC_v(\psi)^\gamma = \UDC_{v \gamma}(\psi)$) for any $u \in \dress$ and $\phi \in \fieldsF$ (resp. any $v \in \udress$ and $\psi \in \invF$).
\end{lemma}

This proves, as expected, that generically $\fieldsF$ and $\invF$ are not isomorphic as field spaces.

\begin{proof}
It is easy to check that the inverse map for $\DC_u$ is $\UDC_{\iota(u)}$. These maps cannot be $\ggG$-equivariant since the equivariance of $\DC$ (resp. $\UDC$) requires to change at the same time $\phi$ and $u$ (resp. $\psi$ and $v$) as  seen in Prop.~\ref{prop GT-DC-UDC}: here, one has $\DC_u(\phi^\gamma) = \GA(\GA(\phi, \gamma), u) = \GA(\phi, \gamma u) = \DC_{\gamma u}(\phi)$ and $\UDC_v(\psi)^\gamma = \GA(\GA(\psi, v), \gamma) = \GA(\psi, v \gamma) = \UDC_{v \gamma}(\psi)$.
\end{proof}

\medskip
The DFM has been formalized in \cite{FourFranLazzMass14a} in terms of fiber bundles from ideas developed in \cite{MassWall10a}. Let us summarize part of this method in the current approach.\footnote{In \cite{FourFranLazzMass14a}, the method was developed in a very general approach: for instance, the symmetry group to be removed was not necessary the whole group $G$, but a subgroup of it. Here we will not consider this situation.} With the previous notations, consider a gauge invariant Lagrangian $L(\phi_1, \dots, \phi_r)$, and suppose there exists (in the model) a natural way to define a field $u \in \dress$. Then, one can perform a change of variables from the field spaces $\fieldsF_i$ to the field spaces $\invF_i$ by using the dressing field map $\DC_u$, which associates to $\phi_i$ the gauge invariant field $\phi_i^u \defeq\DC(\phi_i, u) \in \invF_i$. The Lagrangian can then be written in terms of the $\phi_i^u$'s, on which all the actions of the gauge group $\ggG$ are trivial, so that the $\ggG$ is not relevant anymore in the model and can be thus ignored. It is explained in \cite{MassWall10a, FourFranLazzMass14a} that the so-called unitary gauge in the Electro-Weak sector of the SMPP, whose purpose is to get rid of the $SU(2)$-gauge symmetry, is simply such a change of variables for a natural dressing field in the model.

Notice that the DFM, as a change of variables in the field spaces, is invertible, at least in a formal way, since  one can “undress” all the fields $\phi_i^u$'s using the Un-Dressing Composer $\UDC$ with the undressing field $v = \iota(u)$ (application of Lemma~\ref{lemma DCu and UDCv}). Our “formal” reservation is due to the fact that such an undressing field may not be “natural” to define for a model without symmetry! In fact, some examples of this procedure have been described in the literature, where the Un-Dressing Composer was used to add an “artificial” gauge symmetry in some models where some good candidate for an undressing field $v$ could be proposed. What is a “good candidate” has to be defined in each situation. For instance, it is explained in \cite{FourFranLazzMass14a} how to promote in such a way a Proca-like Lagrangian describing a gauge invariant massive vector field $A_\mu$ to a Stueckelberg Lagrangian which implements a $U(1)$-gauge symmetry.

Thus, beside the two field-composers $\GT$ (gauge transformations) and $\DC$ (dressing), the Un-Dressing Composer $\UDC$ defined in Prop. \ref{prop GT-DC-UDC} might also have a role to play in gauge field theory.

\section{Gauge Fixing and Dressing Fields}

Let us now show how the formalism introduced in the previous section can be used to revisit the gauge fixing procedure in QFT in the light of the DFM.

\subsection{Gauge Fixing in QFT as an instance of the DFM}
\label{sec gauge fixing as an instance of the DFM}

A \emph{gauge fixing map} is a map $\GF : \fsF[1] \times \cdots \times \fsF[r] \to \fsV$ where the $\fsF[i]$'s are functional spaces underlying the field spaces $\fieldsF_i$ of the model and $\fsV$ is a functional space with values in a vector space $V$. In order to simplify the notations, let us write $\GF(\varphi_i)$ for $\GF(\varphi_1, \dots, \varphi_r)$.

\begin{definition}
\label{def gauge fixing condition}
Given a gauge fixing map $\GF$, its associated \emph{gauge fixing condition} is the gauge fixing equation $\GF(\GA(\phi_i,g)) = 0$ to be solved for $g \in \fsG$ while the $\phi_i \in \fieldsF_i$ are fixed. 
\end{definition}

Notice that, at this point, we only specify the field spaces for the $\phi_i$'s. Concerning $g \in \fsG$, we cannot yet determine its field space and we do not know to which field-composer the gauge action transformation $\GA$ will have to be promoted in the equation $\GF(\GA(\phi_i,g)) = 0$. In order to determine this field space, we will need to identify the action of $\ggG$ to which $g$ is subjected. In order to do that, we consider \emph{ideal gauge fixing maps}, see \emph{e.g.} \cite[right after eq.~(3.327)]{Bert96a}, \cite[p.~361]{AzcaIzqu95a}.

\begin{definition}
\label{def ideal gauge fixing}
An \emph{ideal gauge fixing map} is a gauge fixing map $\GF$ such that, for any $\phi_i \in \fieldsF_i$, there is a \emph{unique} $g \in \fsG$ which solves the gauge fixing equation. 
This implies that there is a well-defined one-to-one map $\GFM :  \fieldsF_1 \times \cdots \times \fieldsF_r \to \fsG$.
\end{definition}

The Gribov ambiguity raises the question of the existence of such ideal gauge fixing maps \cite{Grib78a, Sing78a}. It is out of the scope of the present paper to get involved in that difficult problem. We will adopt the usual “practical” point of view that the gauge fixing maps of interest are ideal.

\smallskip
For an ideal gauge fixing map, we propose the following procedure to determine the action of $\ggG$ on $g$. Let us consider a configuration $(\phi_i)$ with $\phi_i \in \fieldsF_i$. Since the gauge fixing map $\GF$ is ideal, there is a unique $g \defeq \GFM(\phi_i) \in \fsG$ such that $\GF(\GA(\phi_i, g)) = 0$. Let $\gamma \in \ggG$ and let us use the notation $\phi_i' \defeq \GT(\phi_i,\gamma) = \GA(\phi_i, \gamma)$ (these are true gauge transformations). Then there is a unique $g' \defeq \GFM(\phi_i') \in \fsG$ such that $\GF(\GA(\phi_i', g')) = 0$. 

\begin{assumption}
\label{assum equivariant map ideal}
If that makes sense (\textit{i.e.} if it is an action), we \emph{define the action of $\gamma$ on $g$} as the map $g \mapsto g'$, so that, with our usual notations, $\GA(g, \gamma) = g^\gamma \defeq g'$. In other words, the action is such that \emph{the map $\GFM$ is $\ggG$-equivariant}, namely, $\GFM(\phi_i^\gamma) = \GFM(\phi_i)^\gamma$.
\end{assumption}

\smallskip
We can now establish the main result of our approach:
\begin{proposition}
The field space of the element $g \in \fsG$ which solves the ideal gauge fixing condition $\GF(\GA(\phi_i,g)) = 0$ is $\dress$.
\end{proposition}

From this proposition we can now deduce that the gauge action transformation $\GA$ in the previous formulation is the field-composer $\DC$, that $\GF$ is promoted to a map $\invF[1] \times \cdots \times \invF[r] \to \fsV$, that the gauge fixing condition looks like $\GF(\DC(\phi_i,u)) = 0$ to be solved for $u \in \dress$  while the $\phi_i \in \fieldsF_i$ are fixed, and that $\GFM :  \fieldsF_1 \times \cdots \times \fieldsF_r \to \dress$. Beware that $\GFM$ looks like a field-composer but the locality is not secured as it will be shown in some examples below.

\begin{proof}
The proof is quite straightforward: one has to solve for $g' \in \fsG$ the equation $\GF(\GA(\phi_i', g')) = 0$. Notice that $\GA(\phi_i', g') = \GA(\GA(\phi_i, \gamma), g') = \GA(\phi_i, \gamma g')$. Since the gauge fixing map is ideal, if $g \in \fsG$ is the unique solution of $\GF(\GA(\phi_i, g)) = 0$, then one must have $g = \gamma g'$, which implies $\GA(g, \gamma) = g^\gamma = g' = \gamma^{-1} g$. This is the action of $\ggG$ on $\fsG$ defining the dressing field space $\dress$.
\end{proof}

Notice that we depart from the usual way to look at the gauge fixing procedure, in which $g$ is considered as an element of the gauge group. This usual identification may have its root in the fact that the most obvious action used in gauge fields theories on the functional space $\fsG$ is the one defining the gauge group $\ggG$. Indeed, the gauge fixing  condition is a local (and possibly non linear) differential equation in terms of the local functions $(\phi_i, g)$ to be solved for $g$. But, isolated from any other (formal) considerations, this equation alone does not tell us which field space $g$ must belong to, since its structure only constrains the functional space $\fsG$. It is then quite natural to implicitly assume that the action to which $g \in \fsG$ is subjected is the one defining $g$ as a local version of an element of the gauge group $\ggG$. All reasonable physicists are inclined to associate such a map $g$ with a (local) gauge transformation. 

Our result challenges this approach since we use a natural criterion to determine the action of $\ggG$ on the functional space $\fsG$. Obviously, the requirement in Assumption~\ref{assum equivariant map ideal}, that $\GFM$ be $\ggG$-equivariant, could be questioned. However, we consider this condition to be the simplest one which respects the spirit of gauge fields theories, where the gauge group is the central object from which it is natural to define the other structures. It is difficult to ask for another natural condition for $\GFM$ which could take into account the actions of the gauge group.

\subsection{The FPGFP revisited}
\label{sec FPGFP revisited}

Let us now show how the Faddeev-Popov method adapts to our framework. Let us use $\phi$ for all the fields in the model, including the gauge potential, the scalar fields and the fermion fields. Denote by $S(\phi)$ the action functional, $\ggG$ the gauge group and $\fieldsF$ the field space of all the fields in the model. 

The usual method requires three hypotheses:
\begin{enumerate}
\item The action functional $S(\phi)$ is gauge invariant.\label{item invariant action}
\item In our context, the integration along the field space $\fieldsF$ can be commuted with the integration along the gauge group $\ggG$.\label{item commutation}
\item Denote by $\dd[\phi]$ the measure on $\fieldsF$, then for any functional $P(\phi)$, one has, for any gauge transformation $\gamma \in \ggG$, 
\begin{align}
\int_\fieldsF \dd[\phi] P(\phi^\gamma) = \int_\fieldsF \dd[\phi] P(\phi).
\label{eq usual invariant measure}
\end{align}
This relation is equivalent to the requirement that the measure $\dd[\phi]$ is gauge invariant since $\int_\fieldsF \dd[\phi] P(\phi^\gamma) = \int_\fieldsF \dd[\phi^{\gamma^{-1}}] P(\phi)$.
\label{item invariant measure}
\end{enumerate}

\bigskip
The FPGFP relies on the “trivial” expression
\begin{align}
\label{eq 1=Delta delta}
1 = \int_\ggG \dd[\gamma] \Delta_{\text{FP}}[\phi, \gamma] \delta( \GF(\phi^\gamma) )
\end{align}
where $\Delta_{\text{FP}}[\phi, \gamma]$ is the functional determinant of the functional derivative of $\gamma \mapsto \GF(\phi^\gamma)$ along $\gamma$ (at fixed $\phi$). In the notations given in Appendix~\ref{sec functional differential and jacobians}, this is the determinant of the linear map $\fsdd_\ggG (\GF \circ \GT)[\phi, \gamma] : T_\gamma \ggG \to T_{\GF(\phi^\gamma)} \fsV$. One can forget for a while about field spaces and look only at the underlying functional spaces. Then one has to compute the linear map $\fsdd_{\fsG} (\GF \circ \GA)[\phi, \gamma] : T_\gamma \fsG \to T_{\GF(\phi^\gamma)} \fsV$. Let $t \mapsto \gamma(t)$ be a smooth curve in $\fsG$ such that $\gamma(0) = \gamma$ and $\dot{\gamma}(0) = \Tgamma \in T_\gamma \fsG$. Then,  for $\psi \defeq \GA(\phi, \gamma)$, one has $\GF \circ \GA(\phi, \gamma(t)) = \GF \circ \GA(\GA(\psi, \gamma^{-1}), \gamma(t)) = \GF \circ \GA(\psi, \gamma^{-1} \gamma(t))$ so that $\fsdd_{\fsG} (\GF \circ \GA)[\phi, \gamma](\Tgamma) = \frac{d}{dt} (\GF \circ \GA)(\psi, \gamma^{-1} \gamma(t))_{\mid t=0} = \fsdd_{\fsG} (\GF \circ \GA)[\psi, e] \circ T_\gamma L_{\gamma^{-1}} (\Tgamma)$. Taking the determinant one gets
\begin{align}
\label{eq Delta phi gamma Delta psi e}
\Delta_{\text{FP}}[\phi, \gamma]
&= \Delta_{\text{FP}}[\psi, e] \Det(T_\gamma L_{\gamma^{-1}})
= \Delta_{\text{FP}}[\psi] \Det(T_\gamma L_{\gamma^{-1}})
\end{align}
where $e \in \fsG$ is the unit element and $\Delta_{\text{FP}}[\psi]  \defeq \Delta_{\text{FP}}[\psi, e]$. This relation is not often mentioned in the literature: it can be found for instance in a similar form as \cite[eq.~(15.5.17)]{Wein05b}.\footnote{In \eqref{eq 1=Delta delta}, the $\delta$ function selects a unique $\gamma_0$ such that $\psi_0 \defeq \GA(\phi, \gamma_0)$ satisfies $\GF(\psi_0) = 0$. Then, using \eqref{eq Delta phi gamma Delta psi e} for this $\gamma_0$, \eqref{eq 1=Delta delta} gives $\Delta_{\text{FP}}[\psi_0]^{-1} =  \Det(T_{\gamma_0} L_{{\gamma_0}^{-1}}) \int_\ggG \dd[\gamma] \delta( \GF(\phi^\gamma) )$. Up to the missing factor $\Det(T_{\gamma_0} L_{{\gamma_0}^{-1}})$, this relation is often used in the literature as a definition of $\Delta_{\text{FP}}[\phi]$.}

One can now insert \eqref{eq 1=Delta delta} into 
\begin{align}
\calZ &\defeq \int_\fieldsF \dd[\phi] e^{i S(\phi)} 
= \int_\fieldsF \dd[\phi] \int_\ggG \dd[\gamma] \Delta_{\text{FP}}[\phi, \gamma]\delta( \GF(\phi^\gamma) )  e^{i S(\phi)}
\nonumber
\\
&= \int_\fieldsF \dd[\phi] \int_\ggG \dd[\gamma] \Delta_{\text{FP}}[\phi, \gamma]\delta( \GF(\phi^\gamma) ) e^{i S(\phi^\gamma)} \quad\text{by item \ref{item invariant action}}
\nonumber
\\
&=\int_\ggG \dd[\gamma]  \int_\fieldsF \dd[\phi] \Delta_{\text{FP}}[\phi, \gamma]\delta( \GF(\phi^\gamma) ) e^{i S(\phi^\gamma)} \quad\text{by item \ref{item commutation}}
\nonumber
\\
\label{eq Z FP usual}
&=\int_\ggG \dd[\gamma] \Det(T_\gamma L_{\gamma^{-1}}) \int_\fieldsF \dd[\psi] \Delta_{\text{FP}}[\psi]\delta( \GF(\psi) ) e^{i S(\psi)} \quad\text{by item \ref{item invariant measure} and \eqref{eq Delta phi gamma Delta psi e},  with $\psi = \phi^\gamma$}.
\end{align}
The steps that usually follow in the FPGFP will not concern us.

\medskip
We adapt the usual hypotheses to our framework in the following way. Let $\fsF$ be a functional space. We introduce 5 hypotheses:
\begin{enum-hypothesis}
\item The action functional $S$ and the gauge fixing map $\GF$ are defined on $\fsF$.\label{hyp definition on functional spaces}

\item For any $g \in \fsG$ and $\varphi \in  \fsF$, one has $S \circ \GA(\varphi, g) = S(\varphi)$.\label{hyp invariant action}

\item The measure on any field space is the measure on the underlying functional space. Denote by $\dd[\varphi ]$ the measure on $\fsF$, then for any functional $\varphi \mapsto P(\varphi)$ on $\fsF$, one has, for any $g \in \fsG$,
\begin{align}
\label{eq invariant measure}
\int_{\fsF} \dd[\varphi] P \circ \GA(\varphi, g) = \int_{\fsF} \dd[\varphi] P(\varphi).
\end{align}
\label{hyp invariant measure}

\item In our context, the integration along $\fsF$ can be commuted with the integration along $\fsG$.\label{hyp commutation}
\end{enum-hypothesis}

\ref{hyp definition on functional spaces} is not a strong restriction, since $S$ is a local expression on $\fieldsF$, which turns out to be a local expression on $\fsF$ (by forgetting about the action of $\ggG$). Concerning $\GF$, this was already in its definition.
 
\ref{hyp invariant action} is equivalent to the usual hypothesis~\ref{item invariant action} if one goes from field spaces to functional spaces. This is possible since the invariance of the action is proved in a formal way which only involves the “functional form” of the gauge action, which is encoded into the gauge action transformation $\GA$ appearing in our hypotheses. 

\ref{hyp invariant measure} means that the measure on a field space is not related to its defining action, but only on the underlying functional space equipped with the Gauge Action Transformation $\GA$ as required in \eqref{eq invariant measure}.
Finally, Eq.~\eqref{eq invariant measure} can be related to the standard Eq.~\eqref{eq usual invariant measure}. Indeed,  the fulfillment of \eqref{eq usual invariant measure} can be performed at the functional level, 
requiring only the explicit form of the group action. The functional equivalent of \eqref{eq usual invariant measure} can be written as $\int_{\fsF} \dd[\varphi] P(\varphi^g) = \int_{\fsF} \dd[\varphi] P(\varphi)$, which is \eqref{eq invariant measure} since $\varphi^g = \GA(\varphi,g)$.

Because of \ref{hyp invariant measure}, the measures on $\ggG$ and $\fieldsF$ are the measures on $\fsG$ and $\fsF$ respectively, so that \ref{hyp commutation} is equivalent to the usual hypothesis~\ref{item commutation}.
  
Our hypotheses, written at the level of functional spaces, are also true on field spaces and field-composers on field spaces, when these expressions make sense. For instance, in the following we will use $\DC$ in place of $\GA$ for the proper field spaces.

\medskip
Let us now write the FPGFP in our framework. \ref{hyp definition on functional spaces} will allow to consider $S \circ \DC(\phi,u)$ and $\GF \circ \DC(\phi, u)$ for any $\phi \in \fieldsF$ and $u \in \dress$, since $\DC(\phi,u) \in \fsF$ (since $\DC(\phi,u) \in \invF$).

Because the gauge fixing map $\GF$ is ideal, one has
\begin{align*}
\int_{\dress} \dd[u] \Delta_{\text{FP}}[\phi, u] \delta( \GF \circ \DC(\phi, u) ) = 1
\end{align*}
where $\Delta_{\text{FP}}[\phi, u]$ is the determinant of the functional derivative $\fsdd_{\dress} (\GF \circ \DC)[\phi, u]$, as in the usual method. We insert this equality into the expression we want to evaluate:
\begin{align}
\label{eq Z def}
\calZ &\defeq \int_{\fieldsF} \dd[\phi] e^{i S(\phi)}
= \int_{\fieldsF} \dd[\phi] \int_{\dress} \dd[u] \Delta_{\text{FP}}[\phi, u] \delta( \GF \circ \DC(\phi, u) ) e^{i S(\phi)}
\\
&= \int_{\fieldsF} \dd[\phi] \int_{\dress} \dd[u] \Delta_{\text{FP}}[\phi, u] \delta( \GF \circ \DC(\phi, u) ) e^{i S \circ \DC(\phi, u)} \quad\text{by \ref{hyp invariant action}} \nonumber
\\
&= \int_{\dress} \dd[u] \int_{\fieldsF} \dd[\phi] \Delta_{\text{FP}}[\phi, u] \delta( \GF \circ \DC(\phi, u) ) e^{i S \circ \DC(\phi, u)} \quad\text{by \ref{hyp commutation}} \nonumber
\\
\label{eq Z FP u psi}
&= \int_{\dress} \dd[u] \Det(T_u L_{u^{-1}})\! \int_{\invF} \dd[\psi] \Delta_{\text{FP}}[\psi] \delta( \GF (\psi) ) e^{i S(\psi)} \quad\text{by \ref{hyp invariant measure} and \eqref{eq Delta phi gamma Delta psi e} with $\psi = \DC(\phi, u)$.} 
\end{align}

At this point, sticking to the usual computation in the FPGFP, one can factor out the integration of $u$ along $\dress$, and consider only the remaining integration on the space of invariant fields. The action functional $S$, initially expressed on $\fieldsF$, is now expressed on $\invF$ after the change of field variables $\DC_u : \fieldsF \to \invF$  where $u$ is the field variable of the first integration along $\dress$ (our \ref{hyp definition on functional spaces} allows to do that). 

Notice that, looking at the previous computations, our hypotheses can be reformulated for field spaces in the following way:
\begin{enum-hypothesis}[label=(Hyp.'~\arabic*)]
\item The action functional $S$ is defined on $\fieldsF$ and $\invF$ in the same functional way.\label{hypp definition on fields spaces}

\item For any $u \in \dress$ and $\phi \in  \fieldsF$, one has $S \circ \DC(\phi, u) = S(\phi)$.\label{hypp invariant action}

\item  The measure on $\invF$ is the push-forward of the gauge invariant measure on $\fieldsF$ by $\DC_u : \fieldsF \to \invF$ for any $u \in \dress$ and \emph{it is independent of $u \in \dress$}. For any functional $\psi \mapsto P(\psi)$ on $\invF$, one has
\begin{align*}
\int_{\fieldsF} \dd[\phi] P \circ \DC(\phi, u) = \int_{\invF} \dd[\psi] P(\psi), \text{ for any } u \in \dress.
\end{align*}
\label{hypp invariant measure}

\item In our context, the integration along $\fieldsF$ can be commuted with the integration along $\dress$.\label{hypp commutation}

\end{enum-hypothesis}
In \ref{hypp invariant measure}, the measure on $\invF$ is precisely defined. Let us show that this measure does not depend on the dressing field $u \in \dress$ used to define it through the push-forward \emph{if and only if} item~\ref{item invariant measure} of the usual hypotheses holds. 

Let us fix $u \in \dress$. By its very definition, the push-forward measure $\dd_u[\psi]$ defined on $\invF$ along the map $\DC_u$ from the measure $\dd[\phi]$ on $\fieldsF$ is such that, for any functional $P$ on $\invF$, one has $\int_{\invF} \dd_u[\psi] P(\psi) = \int_\fieldsF \dd[\phi] P \circ \DC_u(\phi)$. Any other dressing field $u' \in \dress$ is related  to $u$ by $u' = \gamma u$ for a unique $\gamma \in \ggG$. Using Lemma~\ref{lemma DCu and UDCv} and \eqref{eq usual invariant measure}, one gets 
\[
\int_{\invF} \dd_{u'}[\psi] P(\psi) = \int_\fieldsF \dd[\phi] P \circ \DC_{\gamma u}(\phi) = \int_\fieldsF \dd[\phi] P \circ \DC_{u}(\phi^\gamma) = \int_\fieldsF \dd[\phi] P \circ \DC_{u}(\phi) = \int_{\invF} \dd_u[\psi] P(\psi)
\]
so that $\dd_{u'}[\psi] = \dd_u[\psi]$. 

Conversely, if the measure $\dd_u[\psi]$ defined on $\invF$ fulfills $\dd_{u'}[\psi] = \dd_u[\psi]$ for any $u, u' \in \dress$, then, on account of previous notations, set $Q(\phi) = P \circ \DC_{u}(\phi)$. The above computation and the hypothesis on $\dd_u[\psi]$ then show that $\int_\fieldsF \dd[\phi] Q(\phi) = \int_\fieldsF \dd[\phi] Q(\phi^\gamma)$ for any $\gamma \in \ggG$, which is item~\ref{item invariant measure} since $Q$ can be any functional ($P \mapsto Q$ is invertible using $\UDC_{u^{-1}}$).

\medskip
In the usual approach, the action functional is gauge invariant and it is always evaluated on the same field space. On the contrary, in our framework, thanks to the functional gauge invariance \ref{hyp invariant action} or to \ref{hypp invariant action}, it is first composed with the dressing field map $\DC_u$ and then expressed on $\invF$.

\bigskip
Since the Lorenz gauge is mainly used in standard Faddeev-Popov calculations, let us consider the gauge fixing map $\GF(A) \defeq \partial^\mu A_\mu$ for any $A \in \connA$ (here $\fsV =\underline{\Lie G}$). It is well-known that this gauge fixing map is not ideal, but let us assume nevertheless that it is,  as assumed in many physical developments, as already quoted \cite[right after eq.~(3.327)]{Bert96a}, \cite[p.~361]{AzcaIzqu95a}. Then, the gauge fixing condition $\GF(\DC(A, u)) = 0$ takes the form of a non linear second order differential equation to be solved for $u \in \dress$:
\begin{align*}
u^{-1} (\partial^\mu \partial_\mu u) + (\partial^\mu u^{-1}) (\partial_\mu u) + (\partial^\mu u^{-1}) A_\mu u + u^{-1} A_\mu (\partial_\mu u) + u^{-1} (\partial^\mu A_\mu) u
= 0
\end{align*}
It is well-known that the solution is a \emph{non local} expression $u(A) = \GFM(A)$, that is, it is expressed in terms of $A$ and (at least symbolically) an infinite number of derivatives of $A$. So, for the Lorenz gauge fixing map, the map $\GFM$ defined in Def.~\ref{def ideal gauge fixing} is non local.

This differs from the usual examples illustrating the DFM \cite{MassWall10a, FourFranLazzMass14a, Fran14a, FranLazzMass15a, FranLazzMass15b, FranLazzMass16a, AttaFranLazzMass18a, Fran21a} where the dressing field $u$ was always defined in a \emph{local} way in terms of the fields in the model. This locality plays a crucial role in the debate between the artificiality \textit{versus} the substantiality of gauge symmetries \cite{Fran19a} (see also \cite[Chap.~5]{BergFranFrieGome23a}). 

We will see in Section~\ref{sec rxi gauge fixing and unitary gauge} that the non locality of $u$ is also a characteristic of the $R_\xi$ gauge fixing map, and that it disappears in the limit $\xi \to \infty$ (the so-called unitary gauge fixing condition).

\subsection{Gauge Fixing in QFT as a change of field variables}
\label{sec gauge fixing as a change of field variables}

The previous interpretation of the FPGFP as an application of the DFM is not satisfactory from the original viewpoint of the dressing approach, which consists in a mere change of field variables. Let us see how such a change of field variables can be implemented in $\calZ$ defined in \eqref{eq Z def} in order to compare with the previous version of the FPGFP.

Let us suppose as before that the gauge fixing map $\GF$ is ideal. We will use the following maps: let $\DCGFM : \fieldsF \to \invF$ be defined by $\DCGFM(\phi) \defeq \DC(\phi, \GFM(\phi))$ (see Definition~\ref{def ideal gauge fixing}) and let $\GFDC : \fieldsF \to \fsV$ be defined by $\GFDC(\phi) \defeq \GF \circ  \DCGFM(\phi) = \GF \circ \DC(\phi, \GFM(\phi))$ for any $\phi \in \fieldsF$. 

Then, any $\phi \in \fieldsF$ defines a unique $u = \GFM(\phi) \in \GFM(\fieldsF) \subset \dress$ such that $\GF \circ \DC(\phi, u)  = 0$. This $u$ is used to dress the fields $\phi$ by defining the invariant field $\psi \defeq \DC(\phi, u) = \DCGFM(\phi) \in \GF^{-1}(0) \subset \invF$. One gets a change of field variables $\fieldsF \ni \phi \mapsto (u, \psi) \in \spaceupsi$ where 
\begin{align*}
\spaceupsi &\defeq 
(\GFM \times \DCGFM)(\fieldsF)
= \{ (u,\psi) \in  \dress \times \invF \ \mid \ \exists ! \phi \in \fieldsF \text{ s.t. } u = \GFM(\phi) \text{ and } \psi = \DCGFM(\phi) \}
\\
& \subset \GFM(\fieldsF) \times \GF^{-1}(0) \subset \dress \times \invF.
\end{align*}
Performing this change of field variables in the functional integral defining $\calZ$ gives
\begin{align}
\label{eq Z in u and psi}
\calZ 
&= \int_{\fieldsF} \dd[\phi] e^{i S(\phi)}
= \int_{\spaceupsi} \dd[u] \dd[\psi] J(u, \psi ; \phi)  e^{i S(\psi)}
\end{align}
where $J(u, \psi ; \phi)$ is the functional determinant (the Jacobian) of the change of field variables $C : \spaceupsi \to \fieldsF$, $C(u, \psi) = \phi \defeq \UDC(\psi, u^{-1})$. The computation of $J(u, \psi ; \phi)$ relies on the computation of the functional differential $\fsdd C$ (see Appendix~\ref{sec functional differential and jacobians}). Since the gauge fixing map $\GF : \fsF \to \fsV$ is ideal, the number of degrees of freedom in $\fsV$, that is $\dim V$, is larger than the number of degrees of freedom in $\fsG$, that is $\dim G$. Let us assume that $\dim V = \dim G$, so that there is no over-determination of $u \in \fsG$ by $\GF$. 

The bijective map $\GFM : \fieldsF \to \GFM(\fieldsF) \subset \dress$ satisfies the constrain $\GFDC(\phi) = \GF \circ \DC(\phi, \GFM(\phi)) = 0$ for any $\phi \in \fieldsF$, so that, for any $X \in T_\phi \fieldsF$, one has $\fsdd \GFDC[\phi](X) = 0$. By the composition law, one gets $0 = \fsdd \GFDC[\phi](X) = \fsdd \GF[\DCGFM(\phi)] \left(\fsdd \DCGFM[\phi](X) \right)$ while $\fsdd \DCGFM[\phi](X) = \fsdd_{\fieldsF} \DC[\phi, \GFM(\phi)](X) + \fsdd_{\dress} \DC[\phi, \GFM(\phi)]\left( \fsdd \GFM[\phi](X) \right)$, so that
\begin{align*}
0
 &=\fsdd \GF[\DCGFM(\phi)] \big( \fsdd_{\fieldsF} \DC[\phi, \GFM(\phi)](X) \big)
+ \fsdd \GF[\DCGFM(\phi)] \Big( \fsdd_{\dress} \DC[\phi, \GFM(\phi)]\big( \fsdd \GFM[\phi](X) \big) \Big)
\\
&= \fsdd_{\fieldsF} (\GF \circ \DC)[\phi, \GFM(\phi)](X) 
+ \fsdd_{\dress} (\GF \circ \DC)[\phi, \GFM(\phi)] \circ \fsdd \GFM[\phi](X).
\end{align*}
In the FPGFP, it is assumed that, for any $\phi \in \fieldsF$, $\fsdd_{\ggG} (\GF\circ \GT)[\phi, \gamma] : T_\gamma \ggG  \to T_{\GF\circ \GT(\phi, \gamma)} \fsV$ is invertible, since its determinant is $\Delta_{\text{FP}}[\phi^\gamma]$. This invertibility is a technical property at the level of \emph{functional spaces}, so that it can be assumed in our framework as well. This implies the invertibility of the map $\fsdd_{\dress} (\GF \circ \DC)[\phi, u] : T_u \dress \to T_{\GF \circ \DC(\phi, u)} \fsV$ for any $\phi \in \fieldsF$ and $u \in \dress$ (the hypothesis $\dim V = \dim G$ applies here). This entails 
\begin{align*}
\fsdd \GFM[\phi](X)
&= - \fsdd_{\dress} (\GF \circ \DC)[\phi, u]^{-1} \circ \fsdd_{\fieldsF} (\GF \circ \DC)[\phi, \GFM(\phi)](X).
\end{align*}
This expression gives the functional variation $\fsdd u[\phi]$ of $u = \GFM(\phi)$ along $\phi$ in the change of field variables $\fieldsF \ni \phi \mapsto (u, \psi) \in \spaceupsi$. Now, one can look at the variation $\fsdd \psi[\phi]$ of $\psi = \DC(\phi, \GFM(\phi))$ in this change of field variables. One has
\begin{align*}
\fsdd \psi[\phi](X)
&= \fsdd_{\fieldsF} \DC[\phi, \GFM(\phi)](X)
+ \fsdd_{\dress} \DC[\phi, \GFM(\phi)] \big( \fsdd \GFM[\phi](X) \big).
\end{align*}
The determinant $J(\phi; u, \psi)$ of the linear map $(\fsdd u[\phi], \fsdd \psi[\phi]) : T_\phi \fieldsF \to T_{(u, \psi)} \spaceupsi \subset T_{u}\GFM(\fieldsF) \times T_{\psi} \GF^{-1}(0) \subset T_{u}\dress \times T_{\psi} \invF$ (with $u = \GFM(\phi)$ and $\psi = \DC(\phi, \GFM(\phi))$) is the inverse of the Jacobian $J(u, \psi ; \phi)$ we have to compute in \eqref{eq Z in u and psi}. This determinant depends on the three functional differentials $\fsdd_{\fieldsF} \DC$, $\fsdd_{\dress} \DC$, and $\fsdd \GF$. The two first depend only on the field content of the model (recall that $\DC$ is a gauge-like transformations of the fields) while the last one is the only one which depends on the gauge fixing map $\GF$.

Notice that this approach is computationally impractical since it requires to characterize the space $\spaceupsi \subset \dress \times \invF$, which is not an easy task at first sight, especially if $\GF$ is defined in terms of some differential operator. It requires also to evaluate the Jacobian $J(\phi; u, \psi)$, and then its inverse $J(u, \psi ; \phi)$. In a practical approach, it is easier to rely on the FPGFP which has proved to be very effective. Indeed, for the FPGFP the spaces on which the integration is performed turn out to be field space $\dress \times \invF$, while in the displayed approach the integration must be performed on the subspace $\spaceupsi \subset \dress \times \invF$ which is difficult to characterize. Moreover, the determinant to compute, $\Delta_{\text{FP}}[\psi]$, is quite manageable in the context of QFT when one uses the usual trick of the Berezin integration along Grassmann field variables.

\subsection{Field variables dependence on the gauge fixing map}
\label{sec dependence on the gauge fixing map}

In the two computations of the functional integral presented above for $\calZ$, the main step is the change of field variables $\fieldsF \ni \phi \mapsto (u = \GFM(\phi), \psi = \DCGFM(\phi))\in \dress \times \invF$ defined in Section~\ref{sec gauge fixing as a change of field variables}. It is explicit in \eqref{eq Z in u and psi} but it is only implicit in \eqref{eq Z FP u psi} since the Dirac $\delta$-function selects precisely $\psi = \DCGFM(\phi)$.

This change of field variables depends on the ideal gauge fixing map $\GF$ (in fact, it is \emph{defined} by it through the DFM). Let us understand how this dependence  is carried forward onto the invariant field $\psi$. Let us consider a parametrized family of ideal gauge fixing maps $\GF_\epsilon$ such that $\GF_0 = \GF$. For any $\psi \in \invF$, let us define
\begin{align*}
\kv_{\mid \GF(\psi)} 
&\defeq \frac{d \GF_\epsilon(\psi)}{d \epsilon}_{\mid \epsilon=0} \in T_{\GF(\psi)} \fsV.
\end{align*}
A simple choice for such a family is for instance to consider $\kv \in \fsV$ and $\GF_\epsilon(\psi) = \GF(\psi) + \epsilon \kv$, for which $\kv_{\mid \GF(\psi)} = \kv$ is constant. 

Let us fix $\phi \in \fieldsF$ and define $u_\epsilon \defeq \GFM_\epsilon(\phi)$ and $\ku_{\mid u} \defeq \frac{d u_\epsilon}{d \epsilon}_{\mid \epsilon=0} \in T_u \dress$. Then, one has $\GF_\epsilon \circ \DC(\phi, u_\epsilon)  = 0$ for any $\epsilon$. Upon taking the derivative along $\epsilon$ at $\epsilon = 0$, one gets $\kv_{\mid \GF(\psi)} +  \fsdd_\dress (\GF \circ \DC)[\phi, u] ( \ku_{\mid u} ) = 0$ with $u = u_0$ and $\psi \defeq \DC(\phi, u)$. We assume, as before, that $\fsdd_\dress (\GF \circ \DC)[\phi, u]$ is invertible which yields
\begin{align}
\label{eq u from v by F and DC}
\ku_{\mid u} 
&= - \fsdd_\dress (\GF \circ \DC)[\phi, u]^{-1} \left( \kv_{\mid \GF(\psi)} \right).
\end{align}
Let us define the tangent vector $\xi \defeq T_u L_{u^{-1}} \ku_{\mid u} \in T_e \dress$, and write $u_\epsilon = L_u U_\xi(\epsilon) = u U_\xi(\epsilon)$ where $L_u$ is the left multiplication by $u$ in $\fsG$, $U_\xi(\epsilon)$ is a curve in $\dress$ with $U_\xi(0) = e$, and $\tfrac{d U_\xi(\epsilon)}{d \epsilon}_{\mid \epsilon=0} = \xi$.\footnote{One may think of $U_\xi(\epsilon)$ as the curve $e^{\epsilon \xi}$.} Notice that  under a gauge transformation, $U_\xi(\epsilon)$, and so $\xi$, are invariant since $u$ supports on the left the entire right action of $\ggG$ ($u^{\gamma} = \gamma^{-1}u$).

Let us consider now the family of dressed fields $\psi_\epsilon \defeq \DC(\phi, u_\epsilon)$ (with $\psi = \psi_0$) and let $\delta_\xi \psi \defeq \frac{d \psi_\epsilon}{d \epsilon}_{\mid \epsilon=0} \in T_\psi \invF$. With the previous parametrization, one gets
\begin{align} \label{eq delta xi psi xi}
\delta_\xi \psi 
= - \fsdd_\dress \DC[\phi, u] \circ \fsdd_\dress (\GF \circ \DC)[\phi, u]^{-1} \left( \kv_{\mid \GF(\psi)} \right)
= \fsdd_\dress \DC[\phi, u] \circ T_e L_u ( \xi ).
\end{align}
One can introduce a second parametrization along the gauge group as $\gamma_\epsilon \defeq u_\epsilon u^{-1} = u U_\xi(\epsilon) u^{-1} \in \ggG$.
Then $\Txi \defeq \frac{d \gamma_\epsilon}{d \epsilon}_{\mid \epsilon=0} = \Ad_u \xi \in \Lie \ggG = T_e \ggG$ and $\psi_\epsilon= \DC(\phi, \gamma_\epsilon u) = \DC(\GT(\phi, \gamma_\epsilon), u)$. Under a gauge transformation by $\gamma \in \ggG$, $\Txi$ transforms as $\Txi \mapsto \Ad_{\gamma^{-1}} \Txi$. Denote by $\delta_{\Txi} \phi \defeq \fsdd_\ggG \GT[\phi, e](\Txi)$ the infinitesimal gauge transformation of $\phi$ along $\Txi$. Then one gets another expression for $\delta_\xi \psi $:
\begin{align}
\label{eq delta xi psi Txi}
\delta_\xi \psi 
&= \fsdd_\fieldsF \DC[\phi, u] ( \delta_{\Txi} \phi ).
\end{align}
In \eqref{eq delta xi psi Txi}, since $\delta_{\Txi} \phi$ is an infinitesimal gauge transformation, $\delta_\xi \psi$ can be understood as an infinitesimal version of the dressing at $(\phi,u)$ applied to $\delta_{\Txi} \phi$. In order to fully understand \eqref{eq delta xi psi xi}, let us forget about field spaces and look only at their underlying functional spaces. Then one has to compute the derivative along $\epsilon$ of $\GA(\phi, u U_\xi(\epsilon)) = \GA( \GA(\phi, u), U_\xi(\epsilon)) = \GA( \psi, U_\xi(\epsilon))$, which amounts to 
\begin{align}
\label{eq delta xi psi from xi}
\delta_\xi \psi &= \fsdd_{\fsG} \GA[\psi, e](\xi).
\end{align}
As a \emph{functional relation}, this expression depends only on $\psi$ and $\xi$, and not on $\phi$ and $u$ (this was not obvious at first sight in \eqref{eq delta xi psi xi}). It is  the functional expression of an infinitesimal gauge transformation of $\psi$ along $\xi$. But notice that both $\psi$ and $\xi$ support \emph{trivial actions} of the gauge group (and so of its Lie algebra). Hence, this functional relation cannot be interpreted as a true gauge transformation acting on field spaces.

In other words, $\delta_\xi \psi$ in \eqref{eq delta xi psi from xi} has only an interpretation in terms of the (differential) geometry of functional spaces, but not in terms of the infinitesimal gauge group actions. Nevertheless, using the DFM (and more precisely an infinitesimal version of the dressing), it is still possible to interpret $\delta_\xi \psi$ in terms of field spaces using the true infinitesimal gauge transformation $\delta_{\Txi} \phi$ in \eqref{eq delta xi psi Txi} as remarked before. A similar reasoning in terms of functional spaces yields an equivalent relation to \eqref{eq u from v by F and DC} for $\xi$:
\begin{align}
\label{eq xi from v by F and GA}
\xi &= - \fsdd_{\fsG} (\GF \circ \GA)[\psi, e]^{-1} \left( \kv_{\mid \GF(\psi)} \right).
\end{align}
Once again, this expression depends only on the field variable $\psi$.

The variation $\psi \mapsto \psi + \delta_\xi \psi$ does not affect the action $S(\psi)$ since it is \emph{formally} gauge invariant.

\section{$R_\xi$ Gauge Fixing and Unitary Gauge}
\label{sec rxi gauge fixing and unitary gauge}

It is convenient to change our mathematical conventions on gauge fields into more physical ones, for instance conventions close to \cite{PeskSchr08a}, in order to compare the following developments to the ones in the literature.

Here we consider the situation $G = SU(n)$ or $G = U(1)$. Our mathematical notations rely on the following conventions. Let $\{E_a \}$ be a basis of antihermitean elements in $\kg = \ksu(n)$ ($=\{ \text{antihermitean } n\times n \text{ matrices with zero trace}  \}$)  or $\kg = \ku(1) = i \bbR$, such that $[E_a, E_b] = C_{ab}^c E_c$. A connection $1$-form $A \in  \connA$ (Yang-Mills gauge potential) can be decomposed as $A = A_\mu^a E_a \dd x^\mu = A_\mu \dd x^\mu$ with real fields $A_\mu^a$, so that $A_\mu^\dagger = - A_\mu$. An element $\gamma \in \ggG$ close to the identity can be written as $\gamma = e^{\epsilon} = 1 + \epsilon^a E_a + \calO(\epsilon^2)$ with $\epsilon = \epsilon^a E_a$, so that an infinitesimal gauge transformation takes the form $A_\mu^\epsilon = A_\mu + D_\mu \epsilon + \calO(\epsilon^2)$ where $D_\mu \epsilon = \partial_\mu \epsilon + [A_\mu, \epsilon]$.

To stick to standard physical notations we rely on the following conventions. Let $t_a \defeq i E_a$ be Hermitean elements in $\kg$, so that $[t_a, t_b] = i C_{ab}^c t_c$. Let $g$ be the coupling parameter for the interaction described by $G$, and let $\phyA = \phyA_\mu^a t_a \dd x^\mu = \phyA_\mu \dd x^\mu \defeq i g^{-1} A = g^{-1} A_\mu^a t_a \dd x^\mu$ be the physical gauge field, \textit{i.e.} $\phyA_\mu^a = g^{-1} A_\mu^a$ and $\phyA_\mu^\dagger = \phyA_\mu$. Its gauge field strength is $\phyF_{\mu\nu} = \partial_\mu \phyA_\nu - \partial_\nu \phyA_\mu - i g [\phyA_\mu, \phyA_\nu]$, that is $\phyF_{\mu\nu}^a = \partial_\mu \phyA_\nu^a - \partial_\nu \phyA_\mu^a + g C_{bc}^a \phyA_\mu^b \phyA_\nu^c$. Then a gauge transformation close to the identity can be written as $\gamma = e^{i \alpha^a t_a} = 1 + i \alpha^a t_a + \calO(\alpha^2)$ with $\alpha = \alpha^a t_a$. The gauge transformation of $\phyA$ is $\phyA_\mu^\gamma = \gamma^{-1} \phyA_\mu \gamma + i g^{-1} \gamma^{-1} \partial_\mu \gamma$, so that $\phyA_\mu^\alpha = \phyA_\mu - g^{-1} \phyD_\mu \alpha + \calO(\alpha^2)$ with $\phyD_\mu \alpha = \partial_\mu \alpha -i g [\phyA_\mu, \alpha]$. A gauge field $\phi \in \fieldsE$ is subject to the covariant derivative $\phyD_\mu \phi = \partial_\mu \phi - i g \phyA_\mu^a \eta(t_a) \phi$ where $\eta$ is the representation of $\kg$ on $E$ induced by the representation $\ell$ of $G$ on $E$.

\bigskip
The usual way to relate fields in the $R_\xi$ gauge and fields in the unitary gauge is to take the limit $\xi \to \infty$ at the level of Feynman rules and to identify the corresponding propagators with the ones obtained in the unitary gauge. 

In our framework, the $R_\xi$ gauge and the unitary gauge can be written in terms of dressing fields. Thanks to the DFM, the relation between fields in both gauges is achieved through the limit $\xi \to \infty$ in the spaces of type $\invF$ once all the fields of the original theory are dressed via the field-composer $\DC$.
 The Lagrangians in the two gauges are thus related when taking the limit.

Let us illustrate this point with two situations.

\smallskip
Let us first consider the simple situation of an Abelien Higgs model with $G = U(1)$ defined by the Lagrangian
\begin{align}
\label{eq abelian lagrangian}
L[\phyA, \phi] \defeq 
[(\partial_\mu - i e \phyA_\mu) \phi]^\dagger [(\partial^\mu - i e \phyA^\mu)\phi] - V(\phi) - \tfrac{1}{4} \phyF_{\mu\nu} \phyF^{\mu\nu}
\end{align}
where $\phi \in \fieldsE$ (with $E = \bbC$) is a $\bbC$-valued field (here $t_1 = 1$, $\eta=\Id$, and $g = e$), and $V(\phi) = \tfrac{\mu^2}{2} \phi^\dagger \phi + \tfrac{\lambda}{4} (\phi^\dagger \phi)^2$.

For any non zero real parameter $\xi$ and any $v > 0$, consider the $R_\xi$ gauge fixing map
\begin{align*}
\GF_{\xi, v, e}(\phyA, \phi) \defeq \partial^\mu \phyA_\mu - e v \xi \chi \in \underline{\Lie G}
\end{align*}
where $\phi$ is written as $\phi = \frac{v + h}{\sqrt{2}} e^{i\chi}$, which defines $h$ and $\chi$.
This is usually written as the extra term in the Lagrangian:
\begin{align*}
\GF^L_{\xi, v, e}(\phyA, \phi) \defeq -\tfrac{1}{2 \xi}( \partial^\mu \phyA_\mu - e v \xi \chi)^2
\end{align*}
The gauge fixing condition $\GF_{\xi, v, e}(\phyA^u, \phi^u) = 0$, to be solved for $u$ written as $u = e^{i \alpha}$, gives the equation $\partial^\mu \phyA_\mu - \partial^\mu \partial_\mu \alpha - e v \xi(\chi - \alpha) = 0$ to be solved for $\alpha$, that is:
\begin{align}
\label{eq gauge fixing abelien xi}
(\partial^\mu \partial_\mu - e v \xi) \alpha = \partial^\mu \phyA_\mu - e v \xi \chi
\end{align}
This equation determines a unique\footnote{Thanks to conditions at infinity in the Euclidean space.} solution $\alpha_{ev\xi}(\phyA,\phi)$, and so a unique dressing field $u_{ev\xi}(\phyA,\phi) \in \dress$. As for the Lorenz gauge condition, $\alpha_{ev\xi}(\phyA,\phi)$ is non local in the fields $\phyA$ and $\phi$ since one has to invert the Laplacian operator to write $\alpha_{ev\xi}$ in terms of $\phyA$ and $\phi$.

In \cite{FourFranLazzMass14a}, a unitary dressing field $u$ has been defined such that $\phi = \rho u$ (polar decomposition) where $\rho \defeq \abs{\phi}$. This dressing field was used to dress $\phi$ and $\phyA$ into gauge invariant fields and the Lagrangian written in terms of these dressed fields is the so-called “Lagrangian in the unitary gauge”.

Taking the limit $\xi \to \infty$ in \eqref{eq gauge fixing abelien xi}, makes senses if $v\neq 0$. Then one gets the simpler equation $\alpha_{\infty}(\phyA,\phi) = \chi$, that is $u_{\infty}(\phyA,\phi) = u_{\infty}(\phi) = e^{i\chi}$ for $\phi = \rho e^{i\chi}$, so that, \emph{in the space of dressing fields, $\lim_{\xi \to \infty} u_{ev\xi}(\phyA,\phi) = u_{\infty}(\phi) = u$ is the unitary dressing field}. Notice that this limit simplifies the equation in such a way that $u_{\infty}(\phi)$ is now \emph{local} in terms of $\phi$ (and does not depend anymore on $\phyA$). Moreover, 
$u_{\infty}(\phi)$ does not depend on the choice of $v\neq 0$, as expected.

\smallskip
This procedure extends to a more general situation of non Abelian fields. Let $\phi = (\phi_1, \dots, \phi_{2N})$ be real fields subjected to a real representation $\ell$ of $G = SU(N)$. Denote by $T_a \defeq \eta(t_a)$ the real  antisymmetric generators of this representation so that the covariant derivative is $D_\mu \phi = \partial_\mu \phi + g \phyA_\mu^a T_a \phi$ \cite[Chap.~20]{PeskSchr08a}. Let $\hphyA_\mu \defeq \phyA_\mu^a T_a$ (notice that $\hphyA_\mu^\intercal = - \hphyA_\mu$ where ${}^\intercal$ is the transpose matrix) for which 
$ \hphyA_\mu^\gamma = \gamma^{-1}  \hphyA_\mu \gamma + g^{-1} \gamma^{-1} \partial_\mu \gamma$. Consider the Lagrangian $L[\phyA, \phi] \defeq \tfrac{1}{2} (D_\mu \phi)^\intercal (D^\mu \phi) - V(\phi) - \tfrac{1}{4} \phyF_{\mu\nu} \phyF^{\mu\nu}$. Let $\phi_0$ denote a fixed constant configuration of the $\phi$ field that minimizes $V(\phi)$ and let us use the new field $\varphi$ defined by $\phi \rdefeq \phi_0 + \varphi$. For any $\gamma \in \ggG$, we define the gauge transformed $\varphi^\gamma$ of $\varphi$ as $\varphi^\gamma \defeq \ell_{\gamma^{-1}} (\phi_0 + \varphi) - \phi_0$.

In the expansion of $\tfrac{1}{2} (D_\mu \phi)^\intercal (D^\mu \phi)$, we are interested in terms in $\hphyA$ times $\varphi$. These are $\tfrac{1}{2} g (\partial_\mu \varphi)^\intercal \hphyA^\mu \phi_0 - \tfrac{1}{2} g \phi_0^\intercal \hphyA_\mu (\partial^\mu \varphi) = g (\partial_\mu \varphi)^\intercal \hphyA^\mu \phi_0$. Using integration by parts, this term is $- g \varphi^\intercal (\partial_\mu \hphyA^\mu) \phi_0$ under the integration over space-time. The $R_\xi$ gauge fixing condition is chosen in order to cancel this term. As an extra term in the Lagrangian, it is 
\begin{align*}
\GF^L(\hphyA, \varphi) = -\tfrac{1}{2 \xi} \sum_{a} ( \partial^\mu \phyA^a_\mu - g \xi \varphi^\intercal T^a \phi_0 )^2
\end{align*}
with $T^a \defeq K^{ab} T_b$ for the Killing metric $K$ of $SU(N)$ where $K_{ab}\propto \tr(T_aT_b)$. This extra term is associated to the gauge fixing map defined by 
\begin{align*}
\GF_{\xi, \phi_0, g}(\hphyA, \varphi) \defeq \GF^a_{\xi, \phi_0, g}(\hphyA, \varphi) T_a \in \underline{\Lie G}
\end{align*}
where
\begin{align*}
\GF^a_{\xi, \phi_0, g}(\hphyA, \varphi) \defeq \partial^\mu \phyA^a_\mu - g \xi \varphi^\intercal T^a \phi_0.
\end{align*}
The term $\GF^L(\hphyA, \varphi)$ is nothing but $K(\GF_{\xi, \phi_0, g}(\hphyA, \varphi), \GF_{\xi, \phi_0, g}(\hphyA, \varphi))$ up to a factor which depends on normalizations when the $T_a$'s form an orthogonal basis for $K$.

For any $u \in \dress$, define $\hu \defeq \ell_u \in \underline{GL_{2N}(\bbR)}$. Then, one has to solve for $u \in \dress$ the non linear second order differential equation
\begin{multline}
\label{eq gauge fixing non abelien xi}
g^{-1} \hu^{-1} \partial_\mu \partial^\mu \hu
+ g (\partial_\mu \hu^{-1}) (\partial^\mu \hu)
+ (\partial_\mu \hu^{-1}) \hphyA^\mu \hu
+ \hu^{-1} \hphyA^\mu (\partial_\mu \hu)
+ \hu^{-1} (\partial_\mu \hphyA^\mu) \hu
\\
+ g \xi (\phi_0^\intercal T^a \hu^{-1} \varphi) T_a
+ g \xi (\phi_0^\intercal T^a \hu^{-1} \phi_0) T_a
- g \xi (\phi_0^\intercal T^a \phi_0) T_a
= 0
\end{multline}
Notice that the last term in the LHS is zero since $T^a$ is antisymmetric. In case the gauge fixing map $\GF_{\xi, \phi_0, g}$ is ideal, this equation defines a unique dressing field $u_{\xi, g, \phi_0}(\phyA, \phi) \in \dress$, which is clearly a non local expression in terms of the fields $\phyA$ and $\phi$.

The limit $\xi \to \infty$ ($g$ and $\phi_0$ fixed) of eq.~\eqref{eq gauge fixing non abelien xi} reduces to the simple family of algebraic equations
\begin{align}
\label{eq gauge fixing non abelien xi infinity}
\phi_0^\intercal T^a \hu^{-1} \phi  = 0 \quad \text{for any $a$}.
\end{align}
This system of equations is the one defining the unitary gauge for a very general model of broken local symmetries, see for instance \cite[eq. (3.2)]{Wein73a}. This equation defines a unique dressing field $u_{\infty, \phi_0}(\phi) \in \dress$ which is local in terms of $\phi$ and does not depend on $\phyA$.

So, as for the case of the Abelien Higgs model, the limit $\xi \to \infty$ can be performed in the space of dressing fields $\dress$ as $\lim_{\xi \to \infty}  u_{\xi, g, \phi_0}(\phyA, \phi) = u_{\infty, \phi_0}(\phi)$ and it goes from a non local expression in terms of the fields $\phyA$ and $\phi$ to a local expression in terms of $\phi$ alone.
Notice that \eqref{eq gauge fixing non abelien xi infinity} implies that $u_{\infty, \phi_0}(\phi)$ only depends on the direction of $\phi_0\neq 0$.

\medskip
The above-mentioned limit procedures are not rigorously established from a mathematical point of view. In the Abelian case, one can consider the Fourier transform of the original equation \eqref{eq gauge fixing abelien xi} to get an algebraic equation for which the limit procedure is clear. But for non Abelian fields, it requires more mathematical developments to consider the limit from eq.~\eqref{eq gauge fixing non abelien xi} to eq.~\eqref{eq gauge fixing non abelien xi infinity}. Our heuristic approach should be supported by topological considerations on field spaces (introducing Sobolev norms for instance), which is out of the scope of this paper.

\section{Conclusion}
\label{sec conclusion}

In this paper, we have revisited the DFM within a new mathematical framework tailored to QFT. This framework distinguishes between functional spaces and field spaces, the latter being functional spaces with specific actions of the gauge group according to the model at hand. We have shown that the gauge fixing procedure performed in the functional path integral of QFT is an example of the dressing method. Additionally, we illustrated how the Fadeev-Popov gauge fixing procedure can be reformulated using this new formalism. 
Notably, with  $R_\xi$ gauge fixing conditions and “unitary gauges” now understood in terms of dressing fields, we showed that taking the limit $\xi \to \infty$ can be realized within the space of dressing fields. As an outcome the locality of the dressing field is restored. This provides new insights on the relationship between these two types of gauge.

\section*{Acknowledgments}

We would like to thank Louis Usala for suggestions after reading the manuscript. 
\appendix

\section{Functional Differentials and some Jacobians}
\label{sec functional differential and jacobians}

We present in this Appendix some notations concerning functional differentials adapted to our framework. These definitions have been used in the main text. Here, we use them to compute some Jacobians associated to changes of field variables in the functional integrals that are induced by the DFM in the unitary gauge. Some of these computations have been presented before in \cite{MassWall10a}, but in a less complete manner.

Let $F_i$ be spaces and let $\fsF[i]$ be their associated functional spaces on an open set $U$ of $M$. We will look at $\fsF[i]$ as “infinite dimensional smooth  manifolds” on which it is possible to consider some structures usually defined on ordinary manifolds. Obviously, this would require a lot of work to define precisely smoothness on these spaces and smoothness of maps between these spaces (as considered in the following). It is out of the scope of this paper to do that, since we will only be interested in the algebraic part of the obtained structures, not in their analytic existence.\footnote{In fact, we will consider the geometry of these spaces using an approach quite similar to the one developed and described in \cite{FrolKrie88a, KrieMich97a}.} The only basic structure we formally introduce is the tangent space $T_{f_i} \fsF[i]$ of $\fsF[i]$ at $f_i \in \fsF[i]$, which consists of all the $\dot{\gamma}(0) = \tfrac{d \gamma}{dt}_{\mid t=0}$ for $\gamma : (-\epsilon, \epsilon) \to \fsF[i]$ any smooth curve in $\fsF[i]$ such that $\gamma(0) = f_i$ (here $\epsilon > 0$). This reproduces the usual definition of the tangent space in ordinary differential geometry.

Let $C : \fsF[1] \to \fsF[2]$ be a map between two functional spaces. For any $f_1 \in \fsF[1]$, the linear tangent map of $C$ at $f_i$ is the linear map $\fsdd C[f_1] : T_{f_1} \fsF[1] \to T_{C(f_1)} \fsF[2]$ defined by $\fsdd C[f_1](\dot{\gamma}(0)) \defeq \frac{d C \circ \gamma(t)}{dt}_{\mid t = 0}$ for any smooth curve $\gamma$ as before. Thus, $\fsdd C$ will be called the functional differential of $C$. This again reproduces the usual definition. This definition is also a general version of the “functional derivative” introduced in Field Theory, where the $F_i$ are vector spaces and $\gamma(t) = f_i + t X_i$ for a $X_i \in \fsF[1]$. 

For $C_1 : \fsF[1] \to \fsF[2]$ and $C_2 : \fsF[2] \to \fsF[3]$, one has the composition law (or chain rule) $\fsdd (C_2 \circ C_1)[f_1](X_1) = \fsdd C_2[C_1(f_1)]( \fsdd C_1[f_1](X_1) )$ for any $X_1 \in T_{f_1} \fsF[1]$.

For $C : \fsF[1] \times \fsF[2] \to \fsF[3]$, we denote by $\fsdd_{\fsF[i]} C$, for $i=1,2$,  the functional differentials along the two functional spaces $\fsF[i]$, where $\fsdd_{\fsF[i]} C[f_1,f_2] : T_{f_i} \fsF[i] \to T_{C(f_1, f_2)} \fsF[3]$. The total functional differential is then $\fsdd C[f_1,f_2](X_1, X_2) = \fsdd_{\fsF[1]} C[f_1,f_2](X_1) + \fsdd_{\fsF[2]} C[f_1,f_2](X_2)$ for any $X_i \in T_{f_i} \fsF[i]$. We can write this identity as $\fsdd C = \fsdd_{\fsF[1]} C + \fsdd_{\fsF[2]} C$.

\medskip
Let us consider a change of field variables given by $C : \fsF[1] \to \fsF[2]$ (where $\fsF[1]$, resp. $\fsF[2]$, collects all the initial fields, resp.{} the final fields). The  corresponding Jacobian to be computed in the functional integration is the functional determinant of the linear map $\fsdd C[f_1]$. Such a computation was already proposed in \cite{MassWall10a} for the DFM applied to the Electro-Weak sector of the Standard Model, but there, it was not completely described.

Let us first consider the Abelian case described by the Lagrangian \eqref{L-U(1)}. The original field variables in the functional integral are $(\phyA, \phi) \in \fsbbR^m \times \fsbbC$ (remember that $\phyA_\mu \in i \fsku{1} = \fsbbR$). The dressing field $u \in \fsU{1}$ for the unitary gauge is defined by writing the polar decomposition $\phi = \rho u$ with $\rho \in \fsbbRps$. Let $\phya = \DC(\phyA, u) = \phyA + \tfrac{i}{e} u^{-1} \dd u \in \fsbbR^m$. Then the new variables are $(\phya, \rho, u) \in \fsbbR^m \times \fsbbRps \times \fsU{1}$. For the forthcoming computations, it is convenient to change the variable $\rho\in \fsbbRps$ into the variable $\sigma \in \fsbbR$ by the relation $\rho = e^\sigma$. We then define the mapping $C :  \fsF[1] \defeq \fsbbR^m \times \fsbbR \times \fsU{1} \to \fsF[2] \defeq \fsbbR^m \times \fsbbC$ as $C(\phya, \sigma, u) = (\phyA, \phi) = (\phya + \tfrac{i}{e} u \dd u^{-1}, e^\sigma u)$. Let $\Tphya \in \fsbbR^m$, $\Tsigma \in \fsbbR$ and $\Talpha \in \fsbbR$ and define $\gamma(t) \defeq (\phya + t \Tphya, \sigma + t \Tsigma, u e^{i t \Talpha})$ a curve in $\fsF[1]$ such that $\gamma(0) = (\phya, \sigma, u)$ and $\dot{\gamma}(0) = (\Tphya, \Tsigma, i \Talpha) \in T_{(\phya, \sigma, u)} \fsF[1] \simeq \fsbbR^m \times \fsbbR \times \fsku{1}$. A straightforward computation then gives $\fsdd C[\phya, \sigma, u](\Tphya, \Tsigma, i\Talpha) = (\Tphya + \tfrac{1}{e} \dd \Talpha, (\Tsigma + i \Talpha) e^{\sigma} u)$. The Jacobian for this change of variables is then the functional determinant of the functional operator  written in matrix form acting on the components $(\Tphya, \Tsigma, \Talpha)$:
\begin{align*}
\begin{pmatrix}
\bbbone_m & \bbbzero_{1 \times m} & \tfrac{1}{e} \dd \\
\bbbzero_{m \times 2} & e^{\sigma} u & i e^{\sigma} u
\end{pmatrix}
= \begin{pmatrix}
\bbbone_m & \bbbzero_{1 \times m} & \tfrac{1}{e} \dd \\
\bbbzero_{m \times 1} & \phi_1 & - \phi_2 \\
\bbbzero_{m \times 1} & \phi_2 & \phi_1
\end{pmatrix}
\end{align*}
for $\phi = e^{\sigma} u = \phi_1 + i \phi_2$. This Jacobian has to be composed with the one for the change of variables $\rho \mapsto \sigma = \ln \rho$, which is the determinant of the operator $T_\rho \fsbbRps \ni \Trho \mapsto \Tsigma \defeq \rho^{-1} \Trho \in T_\sigma \fsbbR$. The complete operator to consider for the Jacobian associated to the change of field variables $(\phya, \rho, u) \mapsto (\phyA, \phi)$ is then written in matrix form on the components $(\Tphya, \Trho, \Talpha)$ as
\begin{align*}
\bM = \begin{pmatrix}
\bbbone_m & \bbbzero_{1 \times m} & \tfrac{1}{e} \dd \\
\bbbzero_{m \times 2} & u & i e^{\sigma} u
\end{pmatrix}
= \begin{pmatrix}
\bbbone_m & \bbbzero_{1 \times m} & \tfrac{1}{e} \dd \\
\bbbzero_{m \times 1} & \rho^{-1}\phi_1 & - \phi_2 \\
\bbbzero_{m \times 1} & \rho^{-1}\phi_2 & \phi_1
\end{pmatrix}
\end{align*}
This is a matrix block operator of the form $\bM = \smallpmatrix{\bbbone_m & \bD \\ \bbbzero_{m \times 2} & \bE }$. Its determinant can be evaluated using $\Det \bM = e^{\Tr \ln \bM}$ where the definition of $\ln \bM$ relies on the usual series for $\ln(1 + x)$. Since $(\bM - 1)^n = \smallpmatrix{\bbbzero_m & \bD (\bE - \bbbone_2)^{n-1} \\ \bbbzero_{m \times 2} & (\bE - \bbbone_2)^{n} }$, on the diagonal of $\ln \bM$ one gets $\bbbzero_m$ and $\ln \bE$. Applying the trace and the exponential, one then gets $\Det \bM = e^{\Tr \ln \bE} = \Det \bE$. The operator $\bE$ has the form $\bE_{ab}(x,y) = E_{ab}(x) \delta^{(m)}(x-y)$ so that $\Det \bE = \exp \left[ \delta^{(m)}(0) \int \dd^m x \ln (\det E(x)) \right]$ (see for instance \cite{SalaStra70a}) with $\det E = \abs{\phi} = \rho$ so that
\[
\Det \bM = \exp \left[ i \int \dd^m x\   \delta^{(m)}(0) \ln \rho(x) \right]. 
\]

\smallskip
For the $SU(2)$ group, a similar computation can be performed. The original field variables are $(\phyA, \phi) \in \fsbbR^{3m} \times \fsbbC^2$, the dressing field $u$ is uniquely defined by the decomposition $\phi = \eta u \smallpmatrix{0 \\ 1}$ with $\eta = \norm{\phi} \in \fsbbRps$, and the dressed gauge potential is $\phya = \DC(\phyA, u) = u^{-1} \phyA u + \tfrac{i}{g} u^{-1} \dd u \in \fsbbR^{3m}$. The new variables are then $(\phya, \eta, u) \in \fsbbR^{3m} \times \fsbbRps \times \fsSU{2}$. As before we use the variable $\sigma \defeq \ln \eta \in \fsbbR$, so that $C :  \fsF[1] \defeq \fsbbR^{3m} \times \fsbbR \times \fsSU{2} \to \fsF[2] \defeq \fsbbR^{2m} \times \fsbbC^2$ is given by $C(\phya, \sigma, u) = (\phyA, \phi) = (u \phya u^{-1} + \tfrac{i}{g} u \dd u^{-1}, e^\sigma u \smallpmatrix{0\\ 1})$. The functional differential of $C$ is computed using the curve $\gamma(t) \defeq (\phya + t \Tphya, \sigma + t \Tsigma, e^{i t \Talpha} u)$ in $\fsF[1]$ with, for any $(\Talpha^a)\in \fsbbR^3$,  $\Talpha\defeq \Talpha^a  \tau_a$ where the $\tau_a$'s are the Pauli matrices.
$ \fsbbR^3$ is identified with $\underline{\ksu(2)}$ through $(\Talpha^a) \mapsto i \Talpha$.

 One then gets $\fsdd C[\phya, \sigma, u](\Tphya, \Tsigma, i\Talpha) = (u \Tphya u^{-1} + \tfrac{1}{g} \phyD \Talpha, (\Tsigma + i \Talpha) \phi)$ where as before $\phyD \Talpha = \dd \Talpha - i g [\phyA, \Talpha]$.  Using the explicit expressions for the Pauli matrices, this is the operator written in matrix form in components $(\Tphya, \Tsigma, \Talpha)$ as
\begin{align*}
\begin{pmatrix}
\Ad_u^{(m)} & \bbbzero_{1 \times 3m} & \multicolumn{3}{c}{\tfrac{1}{g} \phyD} \\
\bbbzero_{3m \times 1} & \phi_1 & -\phi_4 & \phi_3 & -\phi_2 \\
\bbbzero_{3m \times 1} & \phi_2 & \phi_3 & \phi_4 & \phi_1 \\
\bbbzero_{3m \times 1} & \phi_3 & -\phi_2 & -\phi_1 & \phi_4 \\
\bbbzero_{3m \times 1} & \phi_4 & \phi_1 & -\phi_2 & -\phi_3
\end{pmatrix}
\end{align*}
with $\phi = \smallpmatrix{\phi_1 + i \phi_2 \\ \phi_3 + i \phi_4}$ and where $\Ad_u^{(m)}$ acts as $\Ad_u$ on the $m$ $\ksu(2)$-valued fields $\Tphya_\mu$. As before, one has to compose with the operator associated to the change of field variables $\eta \mapsto \sigma = \ln \eta$. One thus gets the complete operator to consider for the Jacobian associated to the change of field variables $(\phya, \eta, u) \mapsto (\phyA, \phi)$:
\begin{align*}
\bM = \begin{pmatrix}
\Ad_u^{(m)} & \bbbzero_{1 \times 3m} & \multicolumn{3}{c}{\tfrac{1}{g} \phyD} \\
\bbbzero_{3m \times 1} & \eta^{-1} \phi_1 & -\phi_4 & \phi_3 & -\phi_2 \\
\bbbzero_{3m \times 1} & \eta^{-1} \phi_2 & \phi_3 & \phi_4 & \phi_1 \\
\bbbzero_{3m \times 1} & \eta^{-1} \phi_3 & -\phi_2 & -\phi_1 & \phi_4 \\
\bbbzero_{3m \times 1} & \eta^{-1} \phi_4 & \phi_1 & -\phi_2 & -\phi_3
\end{pmatrix}
\end{align*}
Following the same idea as before, the matrix block structure $\bM = \smallpmatrix{\bE_1 & \bD \\ \bbbzero_{3m \times 4} & \bE_2}$ of this operator gives $\Det \bM = (\Det \bE_1)(\Det \bE_2)$ with $\bE_{i, ab}(x,y) = E_{i, ab}(x) \delta^{(m)}(x-y)$. Finally, one has
\begin{align*}
\Det \bM = \exp\left[ i \int \dd^m x \  3 \delta^{(m)}(0) \ln \eta(x) \right].
\end{align*}
This relation can be compared to \cite[eqs.~(3.8) and (3.9)]{GrosKoge93a}, but with the main difference that in the DFM, the choice of a minimum in the potential $V(\phi)$ can be delayed \emph{after} the change of field variables, so that the VEV $v$ does not enter into the game here. We refer to \cite{AttaFranLazzMass18a, FourFranLazzMass14a} for comments on this aspect of the present approach. 

\bibliography{bibliography}

\end{document}